\documentclass[oribibl]{llncs}
\usepackage{makeidx}  % allows for indexgeneration
\usepackage{blindtext}
\usepackage{enumitem}
\usepackage{amsmath}
\usepackage{amssymb}
\usepackage{bm}
\usepackage{tikz}
\usetikzlibrary{arrows.meta,arrows}
\usetikzlibrary{shapes.geometric}
\usetikzlibrary{positioning}
\usetikzlibrary{backgrounds}
\usepackage{algorithm}
\usepackage{algorithmic}
\usepackage{refcount}

\newcommand{\N}{\mathbb{N}}
\newcommand{\R}{\mathbb{R}}
\newcommand{\CP}{\mathcal{P}}
\newcommand{\CE}{\mathcal{E}}
\newcommand{\CG}{\mathcal{G}}
\newcommand{\CI}{\mathcal{I}}
\newcommand{\CO}{\mathcal{O}}
\newcommand{\CV}{\mathcal{V}}
\newcommand{\CW}{\mathcal{W}}
\newcommand{\frog}{\textsc{Frog}}
\newcommand{\minvc}{\textsc{MinVertexCover}}
\newcommand{\mincoloring}{\textsc{MinColoring}}
\newcommand{\maxminvc}{\textsc{MaxMinVertexCover}}
\begin{document}
\mainmatter              % start of the contributions
\title{Routing Games over Time with FIFO policy}
\titlerunning{Routing Games over Time with FIFO policy}  % abbreviated title (for running head)
\author{Anisse Ismaili}
\authorrunning{Anisse Ismaili} % abbreviated author list (for running head)
\tocauthor{Anisse Ismaili}
%
%School of Techno-Business Administration,
\institute{Ito Takayuki Laboratory, Nagoya Institute of Technology, Japan\\
\email{anisse.ismaili@gmail.com}}

\maketitle              % typeset the title of the contribution

\vspace*{-2mm}
%at least 70 and at most 150 words. 
\begin{abstract}
We study atomic routing games where every agent travels both along its decided edges and through time. The agents arriving on an edge are first lined up in a \emph{first-in-first-out} queue and may wait: an edge is associated with a capacity, which defines how many agents-per-time-step can pop from the queue's head and enter the edge, to transit for a fixed delay. We show that the best-response optimization problem is not approximable, and that deciding the existence of a Nash equilibrium is complete for the second level of the polynomial hierarchy. Then, we drop the rationality assumption, introduce a behavioral concept based on GPS navigation, and study its worst-case efficiency ratio to coordination.
\vspace*{-5mm}
\end{abstract}

\vspace*{-8mm}
\section{Introduction}
\vspace*{-2mm}

Numerous selfish agents use a routing network to take shortest paths that may however congest the paths of others.
\emph{Routing games} model such conflictual systems
by a graph of vertices and edges, 
and every agent decides a path from a source to a sink,
path which cost is congestion-dependent.
Routing games find applications in road traffic \cite{wardrop1952road}, 
as well as in routing packets of data via Internet Protocol \cite{koutsoupias1999worst}.
Founding results\footnote{Static routing games were a crucial testbed for the Price of Anarchy, a concept that bounds a game's loss of efficiency due to selfish  behaviors.} have been obtained on \emph{static} routing games \cite{roughgarden2002bad,roughgarden2005selfish,christodoulou2005price,awerbuch2005price,nisan2007algorithmic,roughgarden2009intrinsic}, where each individual path instantaneously occurs everywhere over its decided edges.
%which is also simultaneously congested along its full length. 
%
Such instantaneousness does not reflect that an agent on one edge of its path
is not elsewhere, and cannot congest other edges.
Routing games \emph{over time}, where every agent travels along its route as well as \emph{through time}, were introduced more recently \cite{koch2009nash,anshelevich2009equilibria}.
Introducing time makes games more complicated: pure-strategy Nash equilibria are often not guaranteed; problems such as computing a best-response or an equilibrium are hard; the Price of Anarchy (PoA) can be large. 

We study asymmetric atomic routing games over integer time-steps
that model congestion with a very natural \emph{first-in-first-out} (FIFO) queuing policy on the edges \cite{werth2014atomic}.
Every edge $e$ has an integer \emph{fixed delay} $d_e$ and an integer \emph{capacity} $c_e$. On an edge, every arriving agent is lined up in the edge's FIFO queue (a discrete list); the capacity defines how many agents-per-time-step can pop from the queue's head and transit through the edge, while the others wait the next time-step. Every agent aims at minimizing the time from source to sink.
%
%Rationality (deciding a best-response) and pure Nash equilibria turn out to be computationally intractable, undermining the rationality assumption.
%Also, instead of rationality, assuming that the agents use a GPS navigation system, the loss of efficiency to the coordinated optimum can be very large.

\textbf{Related Work.}
%
%The references below presume a graph of vertices and edges, 
%and agents who route from a source to a sink by deciding a path which cost is congestion-dependent.
Only pure Nash equilibria (PNE) are considered here.
It is the same (resp. different) source/sink in the symmetric (resp. asymmetric) case.

% {Competitive online multicommodity routing},
% WAOA2006, TCS2009
% Harks, Tobias and Heinz, Stefan and Pfetsch, Marc E
% cited : 8, 9 (TO BE DONE)
\cite{harks2006competitive,harks2009competitive}
studies multicommodity routing problems, where asymmetric commodities are routed sequentially and the cost of edges is load dependent. With affine costs, while in the splittable case the PoA is almost 4, in the unsplittable case, computing a best-response is NP-hard, and the PoA is $3+2\sqrt{2}$.

% {A priority-based model of routing},
% TCS Chicago 2008.
% Farzad, Babak and Olver, Neil and Vetta, Adrian},
% cited : 31 (TO BE DONE)
\cite{farzad2008priority}
observes that ``\emph{a car traversing a road can only cause congestion delays to those cars that use the road at a later time}'' and proposes an asymmetric model where every edge has a priority on agents, agents that are congested only by those with a higher priority on the edge. While a global priority (same fixed priority for every edge) guarantees the existence of a PNE, local priorities do not. Several (matching) bounds are derived on the PoA.

% Nash equilibria and the price of anarchy for flows over time
% SAGT2009, TCS2011
% Koch, Ronald and Skutella, Martin
% cited 22,23 (TO BE DONE)
\cite{koch2009nash,koch2011nash} 
%builds upon the literature of deterministic queuing models in road traffic simulation, to thoroughly introduce game-theoretic aspects in continuous flows over time.
introduces competitive flows over time, by building a non-atomic symmetric model upon the literature about deterministic queuing.
Every edge has a fixed transit delay, and a capacity that bounds above the edge's outflow.  %Assuming FIFO, the waiting time can be formulated on each edge. 
It is shown that a sequence of $\varepsilon$-Nash flows converges (as $\varepsilon\rightarrow 0$) to a Nash flow; and  an iterative algorithm is proposed.
%Assuming constant capacities, edges' rates become piecewise constant, and an iterative algorithm is proposed for Nash flows.
While the evacuation-PoA can be arbitrarily large, the time-PoA is in $O(1)$.

% Equilibria in Dynamic Selfish Routing
% SAGT 2009
% Elliot Anshelevich and Satish Ukkusuri 
% cited : 29 (TO BE DONE)
\cite{anshelevich2009equilibria} proposes a dynamic selfish routing model with non-atomic asymmetric agents. A very general delay function $d_e(x,H_e^t)$ of the demand $x$ and the historic $H_e^t$ is introduced, along with a generalized notion of FIFO, which just states that there are no crossovers. Concurrently to \cite{koch2009nash}, it is shown that in the symmetric case, a PNE always exists and can be computed efficiently. However, in the asymmetric case and under a specification of FIFO where an entering agent waits the previous one's end of transit, it is shown that an equilibrium may not exist, and the PoA is bounded below by the number of vertices. Flow independent delays can be reduced to static flows, providing a PoA bound. %by mean of Cartesian product of vertices and time.

% Competitive Routing over Time
% WINE2009, TCS2011
% Martin Hoefer, Vahab S. Mirrokni, Heiko Röglin, and Shang-Hua Teng
% cited : 16, 12 (TO BE DONE)
\cite{hoefer2009competitive,hoefer2011competitive}
proposes temporal (asymmetric and atomic) network congestion games. Every edge has a speed $a_e\in\R_{>0}$ (latency equals speed times weight of agents being processed), and different local policies are studied.
Under FIFO, an edge processes a unitary agent in time $a_e$, while the other agents wait. 
%So no paper's complexity results may imply the other's.
A guaranteed PNE can be computed efficiently for the unweighted symmetric case, despite the NP-hardness of computing a best-response. In the weighted or asymmetric cases, an equilibrium may not exist.
One could reduce one of our edges $e$ to $c_e\times d_e$ speedy-edges having $a_e=1$, but it is \emph{pseudo}-polynomial.
Conversely, it is also unclear how we could reduce this model to the present one. 
%An other model with constant times but time dependent costs is shown to be related to the static model.
%
%

% Atomic routing in a deterministic queuing model
% ORP 2014
% T.L. Werth, M. Holzhauser, S.O. Krumke
% cited : 8 (TO BE DONE)
Our model is the same as in \cite{werth2014atomic}, an atomic variation over integer time-steps of \cite{koch2009nash}, where every edge has a free-flow delay and a capacity that bounds above the inflow-per-time-step. In \cite{werth2014atomic}, the emphasis is rather on \emph{bottleneck} individual objectives, but also on the \emph{sum} on the edges in the path. 
A PNE may not exist; 
Computing a best-response is NP-complete; 
Verifying a PNE is coNP-complete;
Deciding PNE existence is at least NP-hard.
Also, a bound is provided on the PoA.
%
%In light of the hardness of best-responses, 
\cite{werth2015robust} studies games where agents are robust bottleneck optimizers that only know an interval about the cost of edges and learn the actual cost later.

% Competitive Packet Routing with Priority Lists
% MFCS2016
% Tobias Harks, Britta Peis, Daniel Schmand, and Laura Vargas Koch
% cited : 2 (1 below, 1 HS title)
\cite{harks2016competitive} studies a model similar to \cite{werth2014atomic} and ours, but instead of FIFO, studies local and overall priorities on the edges. (Crossovers may occur.)
Some bounds on the price of stability and PoA are derived.
Computing optimal priority lists is shown APX-hard.
Under local priorities, computing best-responses is NP-hard, as well as computing a PNE.

% Braess's Paradox for Flows over Time.
% SAGT2010, TCS2013
% Macko, Martin and Larson, Kate and L'ubos Steskal
% cited 8,4 (TO BE DONE) 
Furthermore, \cite{macko2010braess,macko2013braess} generalizes Braess's Paradox to the model in \cite{koch2009nash}, and Braess's ratio can be much more severe.
%
%
%
% A competitive strategy for routing flow over time
% ACM SIGecom Exchanges 2011
% Bhaskar, Umang and Fleischer, Lisa and Anshelevich, Elliot
% cited 2 (TO BE DONE)
% \cite{bhaskar2011competitive}
% SKIPPED : NOT SALIENT
%
%
%
% Routing Games over Time 
% PhD 30.04.2012
% Ronald Koch
% cited : 5 (TO BE DONE)
%\cite{koch2012routing} is a well written PhD thesis on Routing Games over Time, which widely and thoroughly covers the area.
%
%
%
%{Atomic Dynamic Network Games},
%{Scarsini, Marco and Schr{\"o}der, Marc and Tomala, Tristan},
%{2014}
% cited : 0
%\cite{scarsini2014atomic}
% Dynamic atomic congestion games with seasonal flows
% arxiv 2016 
% Marco Scarsini,  and Marc Schr{\"o}der,  and Tristan Tomala 
% cited : 2 (TO BE DONE)
%and \cite{scarsini2016dynamic} deal with atomic symmetric congestion games, revisit Braess's Paradox and study periodical fluctuations of demands.
%
%
%
% A Stackelberg strategy for routing flow over time
% GEB 2015
% Bhaskar, Umang and Fleischer, Lisa and Anshelevich, Elliot
% cited : 21 (TO BE DONE)
\cite{bhaskar2015stackelberg} considers a model in the fashion of \cite{koch2009nash} and shows that under a Stackelberg strategy, the time-PoA is $(1-1/e)^{-1}$, and the total-delay-PoA is $2(1-1/e)^{-1}$.
%
%
%
% A Network Game of Dynamic Traffic, 
% EC2017, arxiv201705
% Zhigang Cao, Bo Chen, Xujin Chen, Changjun Wang
% cited : 0
%studies dynamism even further
\cite{cao2017arxiv,cao2017ec} propose an extensive form model where agents take new decisions on each vertex.

\textbf{Results.}
%
%In finite, asymmetric routing games over integer time-steps with edges having fixed-delays and capacities on the inflow-per-time-step, and under FIFO queuing \cite[\emph{sum} objective] {werth2014atomic} with global tie-breaking, 
In this paper's model \cite[\emph{sum} objective] {werth2014atomic}, the \emph{new} results are marked here with a \emph{star}$^\ast$:
%label=Th.\arabic*,start=1,
\vspace*{-0.2cm}
\begin{enumerate}[align=left,labelwidth=1.2cm, leftmargin=1.4cm]
\item[Th.1] A pure-strategy Nash equilibrium may not exist
%\footnote{Theorem \ref{th:nexist} differs from \cite[Prop. 4.2]{anshelevich2009equilibria}, where agents are non-atomic, and where FIFO rather means to wait until the edge's previous transit is finished.}
%\footnote{\label{ft:vsHoefer}Theorems \ref{th:nexist} and \ref{th:NPcomplete} differ from \cite[Sec. 3.1]{hoefer2009competitive} where edges have a speed.}
%\footnote{Seemingly close to \cite[Ex. 2]{werth2014atomic}, our example differs by the use of a sum objective (instead of Bottleneck) and of an \emph{overall} tie-breaking priority.}
\footnote{\cite[App. B.1, Fig. B.11]{werth2014atomic} contains a similar result. For culture and the pleasure of the eyes, we include a close didactic counter-example in our Preliminaries.}.
\item[Th.2$^{~}$] The payoffs are well defined and calculable in polynomial-time
\footnote{A close Dijkstra-style algorithm for local priorities lies in \cite[Prop2.2]{harks2016competitive}.}.
\item[Th.3$^{~\ast}$] The best-response decision problem is NP-complete
\footnote{Theorem \ref{th:NPcomplete} differs from \cite[Sec. 3.1]{hoefer2009competitive} where edges instead only have a speed.}
\footnote{Not addressed in \cite[Appendix B]{werth2014atomic}, but for \emph{bottleneck} \cite[Cor. 4]{werth2014atomic}.}
\footnote{Our reduction from problem \minvc~is the base of Th. \ref{th:APXhard}, \ref{th:coNP} and \ref{th:sigma2}.}.
\item[Th.4$^{~\ast}$] The best-response optimization problem is APX-hard,
\item[Th.5$^{~\ast}$] and it is NP-hard to approximate within $|V|^{\frac{1}{6}-\varepsilon}$, and within ${n^{\frac{1}{7}-\varepsilon}}$.
\item[Th.6$^{~}$] Verification of equilibria is coNP-complete\footnote{Theorem \ref{th:coNP} differs from \cite[Cor. 4]{werth2014atomic}, where the objectives are bottlenecks. \cite[Sec. 7]{werth2014atomic} claims that one can derive NP-\emph{hardness} for sum-objectives. Theorems \ref{th:map} and \ref{th:NPcomplete} enable to obtain coNP-\emph{completeness}, which is then partly novel.}.
\item[Th.7$^{~\ast}$] Existence of equilibria is $\Sigma_2^P$-complete\footnote{As a $\Sigma_2^P$-\emph{completeness} result, Th. \ref{th:sigma2} clearly improves on the NP-\emph{hardness} shown in \cite[App. B.3, Th. 19]{werth2014atomic}.}.
\end{enumerate}
\vspace*{-0.2cm}
%These complexity results all hold for unweighted (unitary) agents and consequently extend to weighted agents.

That best-responses are not approximable, deeply questions the rationality assumption of PNE.
We then introduce a behavioral model 
for vehicles taking decisions by GPS, 
inspired by how navigation assistants work: by retrieving information on the current traffic and recomputing shortest paths in real-time.
On the worst-case efficiency ratio of GPS navigation, to coordination, we found:\vspace*{-0.2cm}
\begin{enumerate}[align=left,labelwidth=1.2cm, leftmargin=1.4cm]
\item[Th.8$^{~\ast}$] Allowing walks\footnote{A walk is an alternating sequence of vertices and edges, consistent with the given (di)graph, and that allows repetitions and infiniteness.} as strategies, GPS-agents may cycle infinitely.
\item[Th.9$^{~\ast}$] The Price of GPS Navigation is in $\Omega(|V|+n)$ as the number of vertices $|V|$ and the number of agents $n$ grow.
\end{enumerate}
\vspace*{-0.2cm}

\textbf{Model Discussion.} 
The positioning of waiting queues on the edges' tails, and of fixed-delays inside edges, is without much loss of generality.
Indeed, this choice reduces in polynomial time from/to models where the queue occurs after the fixed delay, where queues are on the nodes and fixed delays on the edges, where edges are unoriented, etc. 
Furthermore, one can model starting times by adding edges, and bottlenecks by delay $d_e=0$ edges.
One can also note that on each edge $e$,
delay $\lfloor(|q_e|-1)/c_e\rfloor+d_e$ is 
almost an affine function of congestion $|q_e|$ (the queue's length).
Since the unweighted agents case that we consider is a particular case of the weighted case (and Algorithm \ref{alg:map} adapts), our complexity results and efficiency lower bounds extend to weighted agents.

\vspace*{-0.2cm}
\section{Preliminaries}\label{sec:prelim}
\vspace*{-0.2cm}

%Given a finite digraph $G=(V,E)$, defined by a finite set of vertices $V$ and by a finite set of (oriented) edges $E\subseteq V\times V$, we take the convention that a simple-path is a sequence of edges that connects a sequence of vertices with no repetitions while respecting the orientation of edges.

%A finite digraph $G=(V,E)$ is a couple defined by a finite set of vertices $V$ and by a finite set of (oriented) edges $E\subseteq V\times V$. Given a digraph, a simple path $\pi$ is a sequence of edges that connects a sequence of vertices with no repetitions.  

%\{(c_e,d_e)\}_{e\in E}

\begin{definition}
A First-in-first-out Routing Game (\frog) is a non-cooperative finite game characterized by tuple $\Gamma=\left(G=(V,E),(c_e,d_e)_{e\in E},N,(s_i,s_i^{\ast})_{i\in N},\succ\right)$.%\vspace*{-1mm}
\begin{itemize}
\setlength\itemsep{0.3em}
\item $G=(V,E)$ is a finite digraph with vertices $V$ and edges $E\subseteq V\times V$.
\item Given edge $e\in E$, positive number $c_e\in\N_{\geq 1}$ is the capacity of edge $e$, and non-negative number $d_e\in\N_{\geq 0}$ is the fixed delay on edge $e$.
\item Finite set $N=\{1,\ldots,n\}$ is the set of agents.
\item Given agent $i\in N$, vertices $s_i,s_i^{\ast}\in V$ are its source vertex and sink vertex.
\item Strict order $\succ$ on set $N$ is a tie-breaking priority on agents. 
\end{itemize}
\end{definition}

For a given \frog, we introduce the following notations.
For every agent $i$,
its \emph{strategy-set} $\CP_i$ consists of every simple path $\pi_i$ from source vertex $s_i$ to sink vertex $s_i^{\ast}$. 
A \emph{strategy-profile} $(\pi_1,\ldots,\pi_n)\in\CP_1\times\dots\times\CP_n$, which for short we denote in bold by $\bm{\pi}\in\bm{\CP}$, defines a strategy for every agent. 
For a given strategy-profile $\bm{\pi}\in\bm{\CP}$; 
strategy $\pi_i$ is the strategy of agent $i$ therein (a simple path from $s_i$ to $s_i^\ast$); 
\emph{adversary strategy-profile} $\bm{\pi}_{-i}\in\prod_{j\neq i}\CP_j$ consists of all strategies in $\bm{\pi}$ but agent $i$'s;
and given strategy $\pi_i'$, strategy-profile $(\pi_i',\bm{\pi}_{-i}) \in\bm{\CP}$ is obtained from strategy-profile $\bm{\pi}$ by changing strategy $\pi_i$ into $\pi_i'$.

Agents travel both along edges and through \emph{time}.
For an agent $i$, \emph{total delay} $C_i:\bm{\CP}\rightarrow\N_{\geq 0}$ 
is a function of the strategy-profile, defined as follows.
%is an integer valued function of the strategy-profile, defined as follows. 
As depicted in Fig. \ref{fig:fifo},
when agent $i$ arrives on edge $e\in\pi_i$, it lines up in a first-in-first-out (FIFO) queue specific to edge $e$. At each time-step, edge $e$ lets the $c_e$ first agents in the queue enter the edge to transit for $d_e$ time steps. Let function $w_{i,e}:\bm{\CP}\rightarrow\N_{\geq 0}$ be the time spent waiting by agent $i$ in the queue of edge $e$.
It follows that agent $i$'s total delay is defined by equality
\vspace*{-1mm}
\begin{eqnarray*}
C_i(\bm{\pi})  &\quad=\quad & \sum_{e\in\pi_i} w_{i,e}(\bm{\pi}) + d_e.
\end{eqnarray*}
If, on one edge, some agents arrive at the same exact time step, then these synchronous agents are ordered in the edge's queue by tie-breaking priority $\succ$. 
Section \ref{sec:frog-delays} shows how to compute, given a strategy profile, the total delays.

\begin{figure}[t]
\centering
\begin{tikzpicture}
\foreach \x in {1,3,5}
{
	\foreach \y in {1,...,3}
	{
		\node[circle,draw=black,fill=white!80!black,minimum size=8] at (0.4*\x,0.4*\y) {};
	}
}
\node[circle,draw=black,fill=white!80!black,minimum size=8] at (0.8,0.4) {};
\node[circle,draw=black,fill=white!80!black,minimum size=8] at (0.8,0.8) {};
\node[circle,draw=black,fill=white!80!black,minimum size=8] at (1.6,0.4) {};

\draw  (0.2,0.2) rectangle (2.2,1.4);
\draw[dotted] (2.2,1.4) -- (3.2,0.8) -- (2.2,0.2);
\node at (-0.3,0.8) {${c_e}$};
\draw[>-<,thick] (-0.1,0.24) -- (-0.1,1.36);
\node at (1.2,-0.2) {${d_e}$};
\draw[->,thick] (0.2,0.05) -- (2.2,0.05);
\node[circle,dotted,inner sep=2mm,draw=black,fill=white] at (3.2,0.8) {};

\draw[dotted] (-0.6,0.2)--(0.4,0.2);
\draw[dotted] (-0.6,1.4)--(0.4,1.4);

\draw[] (-0.6,0.2) -- (-0.6,1.4);
\draw[] (-1.4,0.2) -- (-1.4,1.4);
\draw[] (-2.2,0.2) -- (-2.2,1.4);
\draw[] (-3.0,0.2) -- (-0.6,0.2);
\draw[dashed,thin] (-3.0,0.6) -- (-0.6,0.6);
\draw[dashed,thin] (-3.0,1.0) -- (-0.6,1.0);
\draw[] (-3.0,1.4) -- (-0.6,1.4);
\draw[->] (-2.2,0.05) -- (-0.6,0.05);
\node at (-1.6,-0.2) {queue};

\node[circle,draw=black,fill=white!80!black,minimum size=8] at (-0.8,1.2) {};
\node[circle,draw=black,fill=white!80!black,minimum size=8] at (-0.8,0.8) {};
\node[circle,draw=black,fill=white!80!black,minimum size=8] at (-0.8,0.4) {};
\node[circle,draw=black,fill=white!80!black,minimum size=8] at (-1.6,1.2) {};
\node[circle,draw=black,fill=white!80!black,minimum size=8] at (-1.6,0.8) {};
\node[circle,draw=black,fill=white!80!black,minimum size=8] at (-1.6,0.4) {};
\node[circle,draw=black,fill=white!80!black,minimum size=8] at (-2.4,0.8) {};
\node[circle,draw=black,fill=white!80!black,minimum size=8] at (-2.4,0.4) {};

\draw[dotted, thick] plot [smooth] coordinates 
{ (-2.6,+1.2) (-2.4,+1.0) (-2.45,+0.6) (-2.2,+0.35)
 (-1.95,+0.6) (-2.05,+1.0) (-1.8,+1.25)
 (-1.8,+1.2) (-1.6,+1.0) (-1.65,+0.6) (-1.4,+0.35)
 (-1.15,+0.6) (-1.25,+1.0) (-1.0,+1.25) 
 (-0.75,+1.0) (-0.85,+0.6) (-0.6,+0.35) (-0.5,+0.35)};
\draw[->] (-2.8,+1.2) -- (-2.6,+1.2);
\draw[->] (-0.6,+0.35) -- (-0.3,+0.35);
\end{tikzpicture}
%\vspace*{-0.3cm}
\caption{On a given edge $e$, agents (gray rounds) are first lined up in a FIFO queue. 
The edge lets the $c_e$ first agents enter at each time-step, to travel for $d_e$ time steps.}\label{fig:fifo}
%\vspace*{-0.2cm}
\end{figure}

A rational agent, given an adversary strategy-profile,
individually optimizes its total delay. 
This rationality assumption induces standard concepts:
\begin{definition}
Given agent $i$ and adversary strategy-profile $\bm{\pi}_{-i}$, strategy $\pi_i$ is a best-response if and only if: 
%$$~~\forall \pi_i'\in \CP_i,~~ C_i(\pi_i,\bm{\pi}_{-i})\leq C_i(\pi_i',\bm{\pi}_{-i})$$
\quad $C_i(\pi_i,\bm{\pi}_{-i}) \quad=\quad \min_{\pi_i'\in \CP_i}\left\{ C_i(\pi_i',\bm{\pi}_{-i})\right\}.$
\end{definition}
\begin{definition}
A pure Nash equilibrium (PNE) is a strategy-profile $\bm{\pi}\in\bm{\CP}$ where
%\vspace*{-0.3cm}
$$
\forall i\in N,\quad
\forall \pi_i'\in \CP_i,\quad
C_i(\bm{\pi})\leq C_i(\pi_i',\bm{\pi}_{-i}).
$$
\end{definition}
In plain words, strategy-profile $\bm{\pi}\in\bm{\CP}$ is a PNE if no agent has an individual incentive to deviate from his current strategy, hence if every agent plays a best-response. To illustrate the definitions above 
we recall (Fig. \ref{fig:counter}) a didactic variation of a known counter-example \cite[Fig. B.11]{werth2014atomic}, which implies Theorem \ref{th:nexist}.

\begin{figure}[!t]
\centering
\begin{tikzpicture}
[square/.style={regular polygon,regular polygon sides=4}]
%Frame
\draw[dotted,white]  (0.0,-2.0) rectangle (12.0,+2.0);
%edges of pursuer
\draw[-{Stealth[scale=1.0]}] (1.5,0) -- node[above] {$1$} (2.5,1);
\draw[-{Stealth[scale=1.0]}] (1.5,0) -- node[above] {$1$} (2.5,-1);
\draw[-{Stealth[scale=1.0]}] (2.5,1) -- node[above] {$10$} (3.5,1);
\draw[-{Stealth[scale=1.0]}] (2.5,-1) -- node[above] {$10$} (3.5,-1);
\draw[-{Stealth[scale=1.0]}] (3.5,1) -- node[above] {$1$} (4,0.5);
\draw[-{Stealth[scale=1.0]}] (3.5,-1) -- node[above left] {$1$} (4,-0.5);
\draw[-{Stealth[scale=1.0]}] (4,0.5) -- node[above right] {$0$} (4.5,0);
\draw[-{Stealth[scale=1.0]}] (4,-0.5) -- node[above left] {$0$}  (4.5,0);
%edges of evader
\draw[-{Stealth[scale=1.0]}] (7.5,0) -- node[above] {$1$} (8.5,1);
\draw[-{Stealth[scale=1.0]}] (7.5,0) -- node[above] {$1$} (8.5,-1);
\draw[-{Stealth[scale=1.0]}] (8.5,1) -- node[above] {$10$} (9.5,1);
\draw[-{Stealth[scale=1.0]}] (8.5,-1) -- node[above] {$10$} (9.5,-1);
\draw[-{Stealth[scale=1.0]}] (9.5,1) -- node[above] {$1$} (10,0.5);
\draw[-{Stealth[scale=1.0]}] (9.5,-1) -- node[above left] {$1$} (10,-0.5);
\draw[-{Stealth[scale=1.0]}] (10,0.5) -- node[above] {$0$} (10.5,0);
\draw[-{Stealth[scale=1.0]}] (10,-0.5) -- node[above left] {$0$} (10.5,0);
%externality edge of evader -> pursuer 
\draw[-{Stealth[scale=1.2]}] plot [smooth] coordinates 
{(8.5,1) (7.5,1.5) (4.5,1.5) (3.5,1)};
\node[] at (6,2.4) {$8$};
\draw[-{Stealth[scale=1.2]}] plot [smooth] coordinates 
{(8.5,-1) (7.5,-1.5) (4.5,-1.5) (3.5,-1)};
\node[] at (6,-1.9) {$8$};
%externality edge of pursuer -> evader 
\draw[-{Stealth[scale=1.0]}] (2,1.5) -- node[above] {$2$} (2.5,1);
\draw[-{Stealth[scale=1.0]}] (2,-1.5) -- node[above left] {$2$} (2.5,-1);
\draw[-{Stealth[scale=1.2]},double] plot [smooth] coordinates 
{(2.5,1) (3.5,2) (8.5,2) (9.5,1)} ;
\node[] at (6,1.8) {$10$};
\draw[-{Stealth[scale=1.2]},double] plot  [smooth]  coordinates 
{(2.5,-1) (3.5,-2) (8.5,-2) (9.5,-1)} ;
\node[] at (6,-1.35) {$10$};
%sources
\node[square,draw=black,fill=black!10,inner sep=-0.0em] at (0.6,0) {$1$};
\node[square,draw=black,fill=black!10,inner sep=-0.0em] at (1.0,0) {$5$};
\node[square,draw=black,fill=black!10,inner sep=-0.0em] at (1.4,0) {$6$};
%sources
\node[square,draw=black,fill=black!10,inner sep=-0.0em] at (6.6,0) {$2$};
\node[square,draw=black,fill=black!10,inner sep=-0.0em] at (7.0,0) {$3$};
\node[square,draw=black,fill=black!10,inner sep=-0.0em] at (7.4,0) {$4$};
%sources
\node[square,draw=black,fill=black!10,inner sep=-0.0em] at (1.8,1.5) {$7$};
\node[square,draw=black,fill=black!10,inner sep=-0.0em] at (1.8,-1.5) {$8$};
%sinks
\node[diamond,draw=black,fill=black!10,inner sep=-0.0em] at (4.7,0) {$1$};
\node[diamond,draw=black,fill=black!10,inner sep=-0.0em] at (4.05,0.75) {$3$};
\node[diamond,draw=black,fill=black!10,inner sep=-0.0em] at (4.05,-0.75) {$4$};
%sinks
\node[diamond,draw=black,fill=black!10,inner sep=-0.0em] at (10.7,0) {$2$};
\node[diamond,draw=black,fill=black!10,inner sep=-0.0em] at (10.05,0.75) {$5$};
\node[diamond,draw=black,fill=black!10,inner sep=-0.0em] at (10.05,-0.75) {$6$};
\node[diamond,draw=black,fill=black!10,inner sep=-0.0em] at (10.3,1.0) {$7$};
\node[diamond,draw=black,fill=black!10,inner sep=-0.0em] at (10.3,-1.0) {$8$};
%little names
%ellipse,draw=black, dotted,thick
\node[] at (2.95,0) {{\bf Pursuer}};
\node[] at (9,0) {{\bf Evader}};
\end{tikzpicture}
\caption{An example of a \frog~where there is no PNE. Single edges have capacity one. Double edges have capacity two. The fixed delays of edges are the numbers displayed above. The sources (resp. sinks) of agents are indicated by squares (resp. diamonds). Choose any tie breaking priority compatible with $2~\succ~3\!\sim\!4~\succ~1~\succ~5\!\sim\!6~\succ~7\!\sim\!8$. Agent one \emph{Pursuer} and two \emph{Evader} have two strategies; the others have only one.}\label{fig:counter}
\end{figure}
\begin{theorem}\label{th:nexist}
In a \frog, there may not exist any PNE \cite{werth2014atomic}.
\end{theorem}
\begin{proof}[Theorem \ref{th:nexist}]
%The proof relies on the example in Fig. \ref{fig:counter}, which is constructed as a pursuer-evader game. 
Recall that in a pursuer-evader game, the two agents have two corresponding strategies; the pursuer prefers to decide the same; the evader prefers to decide differently; and consequently there is no PNE.
In Fig. \ref{fig:counter}, all agents are degenerate\footnote{A degenerate agent only has one strategy, but can still incur and cause externalities.}, but agent one \emph{Pursuer} and agent two \emph{Evader}, who can decide between two paths: \emph{up} or \emph{down}. 
%To show that they play a pursuer-evader game, we show that: 
%The idea is that 
Agents three and four transmit from Evader to Pursuer a positive externality for deciding the same, and agents five to eight, from Pursuer to Evader, a negative externality.

If both agents decide the same, without loss of generality path \emph{up}, then Evader makes agent three wait one time-step, who in turn arrives one step after Pursuer where they could have collided; hence the total delay of Pursuer is $12$.  
Also, Pursuer makes agent five arrive at time step $10$ instead of $9$ on the possible collision point with Evader. Moreover, agent seven also arrives there at time $10$. Consequently, this queue is congested by five and seven and Evader waits one time-step. So Evader's total delay is $13$.
Similarly, one can show that when they decide different strategies, 
Pursuer's total delay is $13$, and Evader's is $12$.
To conclude, Fig. \ref{fig:counter} is a pursuer-evader game, and so has no PNE.\qed 
\end{proof}
\begin{definition}We study this sequence of computational problems.
\begin{enumerate}[align=left,labelwidth=2.8cm, leftmargin=3.5cm]
\itemsep0.2em 
\item[\textsc{Frog/Delays}:] 
Given a \frog ~$\Gamma$ and a strategy-profile $\bm{\pi}$,
compute the total delays $\left(C_1(\bm{\pi}),\ldots,C_n(\bm{\pi})\right)$ of every agent.
\item[\textsc{Frog/Br/Opt}:] 
Given a \frog ~$\Gamma$, an agent $i$, and an adversary strategy-profile $\bm{\pi}_{-i}$,
compute a best-response $\pi_i$ for agent $i$.
\item[\textsc{Frog/Br/Dec}:]  Decision version of \textsc{Frog/Br/Opt}.
Given a \frog ~$\Gamma$, an agent $i$, an adversary strategy-profile $\bm{\pi}_{-i}$, and an integer threshold $\kappa\in\N_{\geq 0}$,
decide whether there exists a strategy $\pi_i$ with cost $C_i(\pi_i,\bm{\pi}_{-i})\leq\kappa$.
\item[\textsc{Frog/NE/Verif}:]
Given a \frog ~$\Gamma$ and a strategy-profile $\bm{\pi}$,
decide whether strategy-profile $\bm{\pi}$ is a PNE.
\item[\textsc{Frog/NE/Exist}:]
Given a \frog, decide whether it admits a PNE.
\end{enumerate}
\end{definition}

%Concerning the length of \frog s as inputs,
%the network's representation requires $\Theta(|V|^2)$ Booleans and numbers; describing how many agents go from where to where, $\Theta(|V|^2)$ numbers; and the tie-breaking order, $\Theta(n)$ numbers. In total, the length of a \frog~ is a function  in $\Theta(|V|^2+n)$, as $|V|$ and $n$ grow.
The representation size of \frog s is a polynomial of numbers $|V|$ and $n$. 
We assume that the following concepts are common knowledge:
decision problem, length function, complexity classes P, ZPP, NP, coNP, $\Sigma_2^P$, $\Pi_2^P$, PH, NPO and APX, polynomial-time reduction, L-reduction, hardness and completeness.

\section{Mapping Strategy-Profiles to Total-Delays}\label{sec:frog-delays}

%In most games, the mapping from strategy-profiles to payoffs is easy to compute. 
%However, under local priorities, this mapping is not well defined, when there are directed cycles of length zero \cite{harks2016competitive}.
Strikingly, the mapping from strategy-profiles to payoffs is not well defined under local priorities, when there are directed cycles of length zero \cite{harks2016competitive}.
Under FIFO priorities with a tie-breaking order, we show the following.

\begin{theorem}\label{th:map}
The mapping from strategy-profiles to total-delays is well defined, 
and there is$^{~(c)}$ a polynomial-time algorithm to compute it: 
$\textsc{Frog/Delays}\in \textsc{P}$.
\end{theorem}

%\vspace*{-0.2cm}

\begin{algorithm}[t]
  \setlength{\itemindent}{1em}
  \addtolength{\algorithmicindent}{1em}
\caption{Algorithm for mapping strategy-profiles to total-delays.}\label{alg:map}
\vspace*{0.2cm}
\textbf{Input:} \frog~$\Gamma$, strategy-profile $\bm{\pi}$.\hspace*{1.1cm}
\textbf{Output:} Total delays $C_1(\bm{\pi}),\ldots,C_n(\bm{\pi})$.\vspace*{0.1cm}\\
\textbf{Variables:} Heap of events $\CE$,~ queues $q_e$ on every edge $e$,~ integer results $C_1,\ldots,C_n$.
\vspace*{0.2cm}
\hrule
\begin{algorithmic}[1]
\STATE $\CE\leftarrow\bigcup_{i\in N}\{\left(0,\theta_{\text{queue}},i,\textbf{first}(\pi_i)\right)\}$
\WHILE{$\CE\neq\emptyset$}
	\STATE{$\left(t,\theta,i,e\right)\leftarrow\mbox{\textbf{top}}(\CE) \quad\AND\quad \mbox{\textbf{pop}}(\CE)$}
	\IF{$\theta = \theta_{\text{queue}}$} 
		\STATE{$\mbox{\textbf{push-back}}(q_e,i)$}
		%\STATE{$q_e\overset{\mbox{push-back}}{\xleftarrow{\hspace*{1.3cm}}} i$}
		\STATE {$\CE\overset{\mbox{insert}}{\xleftarrow{\hspace*{0.8cm}}}\left(t+\lfloor\frac{|q_e|-1}{c_e}\rfloor,\theta_{\text{pop}},i,e\right)$} 
		
	\ELSIF{$\theta = \theta_{\text{pop}}$} 
		\STATE{$\mbox{\textbf{pop}}(q_e)$}
		\IF{$e=\mbox{\textbf{last}}(\pi_i)$} 
			\STATE $C_i\leftarrow t+d_e$
		\ELSE		
			\STATE {$\CE\overset{\mbox{insert}}{\xleftarrow{\hspace*{0.8cm}}}\left(t+d_e,\theta_{\text{queue}},i,\mbox{\textbf{next}}(\pi_i,e)\right)$} 
		\ENDIF
	\ENDIF
\ENDWHILE
\RETURN $C_1,\ldots,C_n$.
\end{algorithmic}
\end{algorithm}
\begin{proof} 
\vspace*{-0.4cm}
This proof relies on Alg. \ref{alg:map} that is made deterministic by order $\succ$. 
 Algorithm \ref{alg:map} sequentially performs events, Dijkstra-like, along time. 
Type $\theta_{\text{queue}}$ events are when the agent is lined-up at the queue's tail.
Type $\theta_{\text{pop}}$ events are when the agent pops from the queue's head and starts the fixed transit delay on edge $e$. 
 An \emph{event}  is a quadruplet $\left(t,\theta,i,e\right)$ that occurs
 on edge $e\in E$
 when agent $i\in N$
  performs an event of type $\theta\in\left\{\theta_{\text{queue}},\theta_{\text{pop}}\right\}$ 
  at time $t\in\N_{\geq 0}$. 
As Alg. \ref{alg:map} iterates, time goes forward. Heap $\CE$ is the set of next events for the agents that did not finish their path. 
The events in heap $\CE$ are ordered according to time $t$ (lowest first), 
then type $\theta$ ($\theta_{\text{queue}}$ before $\theta_{\text{pop}}$),
and then agent $i$'s priority (high priority first).
Then, the overall next event is the top of the heap $\mbox{\textbf{top}}(\CE)$, and can be removed by $\mbox{\textbf{pop}}(\CE)$.
For every edge $e$, queue $q_e$ is maintained as agents line-up using $\mbox{\textbf{push-back}}(q_e,i)$ and pop using $\mbox{\textbf{pop}}(q_e)$. For every path $\pi_i$, we can access first edge $\textbf{first}(\pi_i)$, last edge $\mbox{\textbf{last}}(\pi_i)$, and given an edge $e\in\pi_i,e\neq\mbox{\textbf{last}}(\pi_i)$, we can access the next edge $\mbox{\textbf{next}}(\pi_i,e)$.
 
A queuing event $\left(t,\theta_{\text{queue}},i,e\right)$ develops into the popping event where agent $i$ leaves queue $q_e$ to start the fixed-delay. Crucially, when this queuing event is performed, we know that no further agents can be lined-up prior to agent $i$ in queue $q_e$, since it's the heap's top, which optimizes time and priority. Also, type $\theta_{\text{queue}}$ goes before type $\theta_{\text{pop}}$, so that current congestion is counted. Hence, we know that agent $i$ will spend time $w_{i,e}(\bm{\pi})=\lfloor\frac{|q_e|-1}{c_e}\rfloor$ in the queue. A $\theta_{\text{pop}}$ event generates either the queuing event after delay $d_e$, or total delay $C_i$.

There are at most $2|V|$ events per-agent. Moreover $|\CE|\leq n$. Consequently, Alg. \ref{alg:map} is in polynomial-time $O(n|V|\log_2(n))$.
It was tested in detail under C++.
It adapts to weighted agents by calculation of the weighted length of queues.\qed
\end{proof}

\section{Inapproximability of Best-Responses}

Theorem \ref{th:map} implies that problem \textsc{Frog/Br/Opt} is somewhere inside class NPO, and problem \textsc{Frog/Br/Dec} in class NP. In this section, we show that computing a best-response is hard, and provide two inapproximability results.
\begin{theorem}\label{th:NPcomplete}
Decision problem \textsc{Frog/Br/Dec} is NP-complete.
\end{theorem}
So, a polynomial-time algorithm addressing \textsc{Frog/Br/Dec} is unlikely to exist.
\begin{theorem}\label{th:APXhard}
Optimization problem \textsc{Frog/Br/Opt} is APX-hard.
\end{theorem}
Hence, a PTAS for \textsc{Frog/Br/Opt} would imply a PTAS for every NPO problem that admits a poly.-time constant factor approx. algorithm, which is unlikely.
\begin{theorem}\label{th:inapprox}
For any $\varepsilon\in\R_{>0}$, approximating problem 
\textsc{Frog/Br/Opt} within factor $|V|^{\frac{1}{6}-\varepsilon}$, 
and within factor ${n^{\frac{1}{7}-\varepsilon}}$, are NP-hard.
\end{theorem}
In plain words, it would take an intractable amount of time for an agent to find a path within factor $|V|^{\frac{1}{7}}$ or ${n^{\frac{1}{8}}}$ of the shortest delay.
%a polynomial-time exact algorithm, a PTAS, or even just a $|V|^\textbf{r}$-approximation algorithm are very unlikely to exist.
%A consequence is that agents, even with a powerful computing device, can hardly reach their individual optima, or even just decide a path not too far from the shortest.
A more realistic model may rather drop rationality, and be better based on agents using heuristics.

Before the proofs, a good rule of thumb to distinguish between easy and hard path problems, is whether Bellman's Principle of Optimality is satisfied, or if preference inversions violate the principle\footnote{This assertion is only a rule of thumb, since no state-space is actually precised.}. We introduce a gadget game\footnote{A busy reader can, after a quick look to Fig. \ref{fig:backfire}, jump straight to Lem. \ref{lem:backfire-universal}.}. %modeling a preference inversion that we will widely use.

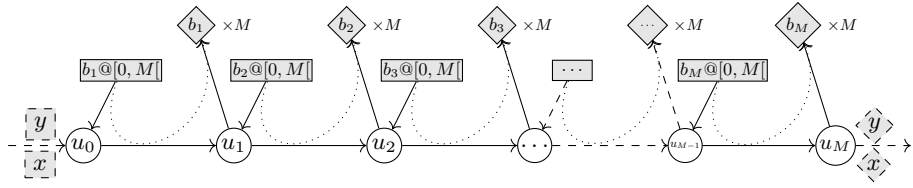
\begin{figure}[!t]
\centering
%\vspace*{-1cm}
\begin{tikzpicture}
%u
\node[circle,draw=black, inner sep=0.1em] (u0) at (0,0) {$u_0$};
\node[circle,draw=black, inner sep=0.1em] (u1) at (2,0) {$u_1$};
\node[circle,draw=black, inner sep=0.1em] (u2) at (4,0) {$u_2$};
\node[circle,draw=black, inner sep=0.1em] (u3) at (6,0) {$\ldots$};
\node[circle,draw=black, inner sep=0.1em,scale=0.5] (u4) at (8,0) {$u_{M-1}$};
\node[circle,draw=black, inner sep=0.1em,scale=0.9] (u5) at (10,0) {$u_M$};
%sources
\node[rectangle,draw=black,fill=black!10, inner sep=0.1em,scale=0.75] 
(v0) at (0.5,1) {$b_1$@$[0,M[$};
\node[rectangle,draw=black,fill=black!10, inner sep=0.1em,scale=0.75] 
(v1) at (2.5,1) {$b_2$@$[0,M[$};
\node[rectangle,draw=black,fill=black!10, inner sep=0.1em,scale=0.75] 
(v2) at (4.5,1) {$b_3$@$[0,M[$};
\node[rectangle,draw=black,fill=black!10, inner sep=0.5em,scale=0.75] 
(v3) at (6.5,1) {$\ldots$};
\node[rectangle,draw=black,fill=black!10, inner sep=0.1em,scale=0.75] 
(v4) at (8.5,1) {$b_M$@$[0,M[$};
%sinks
\node[diamond,draw=black,fill=black!10, inner sep=0.1em,scale=0.75] 
(w1) at (1.5,1.6) {$b_1$};\node[right=0cm of w1,scale=0.67]{$\times M$};
\node[diamond,draw=black,fill=black!10, inner sep=0.1em,scale=0.75] 
(w2) at (3.5,1.6) {$b_2$};\node[right=0cm of w2,scale=0.67]{$\times M$};
\node[diamond,draw=black,fill=black!10, inner sep=0.1em,scale=0.75] 
(w3) at (5.5,1.6) {$b_3$};\node[right=0cm of w3,scale=0.67]{$\times M$};
\node[diamond,draw=black,fill=black!10, inner sep=0.5em,scale=0.5] 
(w4) at (7.5,1.6) {$\ldots$};\node[right=0cm of w4,scale=0.67]{$\times M$};
\node[diamond,draw=black,fill=black!10, inner sep=0.1em,scale=0.75] 
(w5) at (9.5,1.6) {$b_M$};\node[right=0cm of w5,scale=0.67]{$\times M$};
%%% EDGES %%%
\draw[] 
	(u0) edge[->] (u1)
	(u1) edge[->] (u2)
	(u2) edge[->] (u3)
	(u3) edge[->,dashed] (u4)
	(u4) edge[->] (u5);
\draw[] 
	(v0) edge[->] (u0)
	(v1) edge[->] (u1)
	(v2) edge[->] (u2)
	(v3) edge[->,dashed] (u3)
	(v4) edge[->] (u4);
\draw[] 
	(u1) edge[->] (w1)
	(u2) edge[->] (w2)
	(u3) edge[->] (w3)
	(u4) edge[->,dashed] (w4)
	(u5) edge[->] (w5);
\draw[dotted]
	(v0) edge[->,out=250,in=290,looseness=3.1] (w1)
	(v1) edge[->,out=250,in=290,looseness=3.1] (w2)
	(v2) edge[->,out=250,in=290,looseness=3.1] (w3)
	(v3) edge[->,out=250,in=290,looseness=3.1] (w4)
	(v4) edge[->,out=250,in=290,looseness=3.1] (w5);
% limit nodes
\node[rectangle, dashed, draw=black, fill=black!10, above left=-1mm and 2mm of u0](y){$y$};
\node[rectangle, dashed, draw=black, fill=black!10, below left=-1mm and 2mm of u0](x){$x$};
\node[diamond, dashed, draw=black, fill=black!10, inner sep=0.1em,  above right=-0.5mm and 2mm of u5](y){$y$};
\node[diamond, dashed, draw=black, fill=black!10, inner sep=0.1em, below right=-0.5mm and 2mm of u5]{$x$};
\draw[->,dashed] (-1,0) -- (u0);
\draw[->,dashed] (u5) -- (11,0);

\end{tikzpicture}
\caption{\emph{Filter: Agent $y$ is not delayed, but agent $x$ is delayed by at least $M\in\N_{\geq 1}$ steps.} Circles, rectangles and diamonds are resp. vertices, sources and sinks. The idea is that agent $x$ waits one step on every edge $(u_i,u_{i+1})$, but $y$ does not wait. After symbol @ is the source's starting time. An agent $b_i$ starts for every $i\in[1,M]$ and $t\in[0,M[$, hence the number of agents is in $\Theta(M^2)$. The edges have capacity one and fixed-delay zero. Priority $\succ$ satisfies $y\succ b_i\succ x$, for any $i$ (and time) defined in the Figure.}\label{fig:filter}
\end{figure}

\begin{figure}[!t]
%\vspace*{-1cm}
\centering
\begin{tikzpicture}
%\draw[dotted,black]  (0.0,0.0) rectangle (12.0,7.0);
% trigger and time-bomb and the rest
\node[circle, draw=black, inner sep=-0.01em] (w1) at (1,1) {$w_1$};
\node[circle, draw=black, inner sep=-0.01em] (w2) at (2,1) {$w_2$};
\node[circle, draw=black, inner sep=-0.01em] (w3) at (10,1) {$w_3$};
\node[circle, draw=black, inner sep=-0.01em] (w4) at (11,1) {$w_4$};
\draw[->,dotted] (0,1.5) -- (w1);
%\draw[->,dotted] (0,1.0) -- (w1);
\draw[->,dotted] (0,0.5) -- (w1);
\scriptsize 
\node[rectangle,draw=black,dotted, align=center] 
(graph) at (6,1) {Any acyclic digraph, with from\\  $w_2$ to $w_3$ a delay at least one.};
\normalsize
\draw[->,thick] (w1) -- node[below = 0.3cm]{$t$-\textbf{trigger}} (w2);
\draw[dotted] (w2) -- (graph);
\draw[->, dotted] (graph) -- (w3);
\draw[->, thick] (w3) -- node[below = 0.3cm]{\textbf{bomb}} (w4);
\draw[->,dotted] (w4) -- (12,1.5);
%\draw[->,dotted] (w4) -- (12,1.0);
\draw[->,dotted] (w4) -- (12,0.5);
%upper nodes
\node[circle,draw=black, inner sep=-0.01em] (u0) at (2,5) {$u_0$};  
\node[circle,draw=black, inner sep=-0.01em] (u1) at (3.5,5) {$u_1$}; 
\node[circle,draw=black, inner sep=-0.01em] (u2) at (5,5) {$u_2$};
\node[circle,draw=black, inner sep=-0.01em] (u3) at (6.5,5) {$u_3$}; 
\node[circle,draw=black, inner sep=-0.0em,scale=0.6] (u4) at (8.5,5) {$u_{M-2}$}; 
\node[circle,draw=black, inner sep=-0.0em,scale=0.7] (u5) at (10,5) {$u_{M-1}$};  
\node[circle,draw=black, inner sep=-0.0em] (u6) at (11.5,5) {$u_M$};  
%upper nodes bis
\node[circle,draw=black, inner sep=0.2em] (u0b) at (2.5,5) {~};  
\node[circle,draw=black, inner sep=0.2em] (u1b) at (4,5) {~}; 
\node[circle,draw=black, inner sep=0.2em] (u2b) at (5.5,5) {~};
%\node[circle,draw=black, inner sep=0.2em] (u3b) at (7,5) {~}; 
\node[circle,draw=black, inner sep=0.2em] (u4b) at (9.2,5) {~}; 
\node[circle,draw=black, inner sep=0.2em] (u5b) at (10.7,5) {~};  
%upper-1 nodes
\node[circle,draw=black, inner sep=-0.01em] (v1) at (3.5,4) {$v_1$}; 
\node[circle,draw=black, inner sep=-0.01em] (v2) at (5,4) {$v_2$};
\node[circle,draw=black, inner sep=-0.01em] (v3) at (6.5,4) {$v_3$}; 
\node[circle,draw=black, inner sep=-0.01em] (v4) at (8.5,4) {}; 
\node[circle,draw=black, inner sep=-0.0em,scale=0.7] (v5) at (10,4) {$v_{M-1}$};  
\node[circle,draw=black, inner sep=-0.0em] (v6) at (11.5,4) {$v_{M}$};
%upper-right edges.
\draw[->] (u0) -- (u0b); \draw[->] (u0b) -- (u1);
\draw[->] (u1) edge[->] (u1b); \draw[->] (u1b) -- (u2);
\draw[->] (u2) edge[->] (u2b); \draw[->] (u2b) -- (u3);
\draw[->,dashed] (u3) -- (u4);
\draw[->] (u4) edge[->] (u4b); \draw[->] (u4b) -- (u5);
\draw[->] (u5) edge[->] (u5b); \draw[->] (u5b) -- (u6);
%upper-down edges.
\draw[->] (u1) -- (v1);
\draw[->] (u2) -- (v2);
\draw[->] (u3) -- (v3);
\draw[->,dashed] (u4) -- (v4);
\draw[->] (u5) -- (v5);
\draw[->] (u6) -- (v6);
%red-elevator
\scriptsize 
\node[rectangle,draw=black,align=center,rotate=90] (filter) at (2,3.2) 
{Agent $r_{1}$ crosses with\\ delay $0$, but agent $x$\\ is delayed  by $M$.};
\normalsize
\node[below right = 2cm and 0cm of filter,rotate=90]{\textbf{filter}};
\draw[] (w2) -- (filter);
\draw[->] (filter) -- (u0);
%final curved connection
\draw[->] (v1) edge [out=270, in=140] (w3);
\draw[->] (v2) edge [out=270, in=120] (w3);
\draw[->] (v3) edge [out=270, in=100] (w3);
\draw[->,dashed] (v4) edge [out=270, in=90] (w3);
\draw[->] (v5) edge [out=270, in=80] (w3);
\draw[->] (v6) edge [out=240, in=70] (w3);
%
%Agents
%rs and xs
%sources
\node[rectangle,draw=black,fill=black!10,dashed,inner sep=0.1em,left = 0.2cm of w1] 
(xsource) {$x$@?};
\draw[->,dotted] (xsource) -- (w1);
\node[rectangle,draw=black,fill=black!10,above=0.05cm of w1, inner sep=0.1em] 
(r0) {$r_0$@$t$};
\node[rectangle,draw=black,fill=black!10,above=0.05cm of r0, inner sep=0.1em,scale=0.57] 
(r1) {$r_1$@$t+1$};
%sinks
\node[diamond,draw=black,fill=black!10,dashed,right = 0.3cm of w4,inner sep=0.1em] 
(xsink) {$x$};
\draw[->,dotted] (w4) -- (xsink);
\node[diamond,draw=black,fill=black!10,inner sep=0.0em,above right = 0.02cm and 0.02cm of w2] {$r_0$};
%bs
%sources
\node[rectangle,draw=black,fill=black!10,above right=0.2cm and -0.25cm of u0b, inner sep=0.1em,scale=0.67] (b1)  {$b_1$@$t+1$};
\draw[->] (b1) -- (u0b);
\node[rectangle,draw=black,fill=black!10,above right=0.2cm and -0.25cm of u1b, inner sep=0.1em,scale=0.67] (b2)  {$b_2$@$t+1$};
\draw[->] (b2) -- (u1b);
\node[rectangle,draw=black,fill=black!10,above right=0.2cm and -0.25cm of u2b, inner sep=0.1em,scale=0.67] (b3)  {$b_3$@$t+1$};
\draw[->] (b3) -- (u2b);
%%%
\node[rectangle,draw=black,fill=black!10,above right=0.2cm and -0.6cm of u4b, inner sep=0.1em,scale=0.55] (b5)  {$b_{M-1}$@$t+1$};
\draw[->] (b5) -- (u4b);
\node[rectangle,draw=black,fill=black!10,above right=0.2cm and -0.5cm of u5b, inner sep=0.1em,scale=0.67] (b6)  {$b_{M}$@$t+1$};
\draw[->] (b6) -- (u5b);
%%% sink of r1
\node[diamond,draw=black,fill=black!10,inner sep=0.0em,above right = 0.02cm and 0.02cm of u6] {$r_1$};
%sinks of bs
\node[diamond,draw=black,fill=black!10,inner sep=0.0em,above right= 0.1cm and 0.0cm of w4,scale=0.67] (bmsink) {$b_M$};
\node[diamond,dashed, draw=black!50,fill=black!5,inner sep=0.1em,above = 0.0cm of bmsink,scale=0.8] (b2sink) {$\ldots$};
\node[diamond,draw=black,fill=black!10,inner sep=0.0em,above = 0.0cm of b2sink,scale=0.8] (b1sink) {$b_1$};
%
%m-agents sources/links/multiplier/sinks/multiplier
\node[rectangle,draw=black,fill=black!10,above=0.5cm of u1, inner sep=0.1em,scale=0.67] (m1)  {$m_1$@$t+2$};
\draw[->] (m1) -- (u1);
\node[scale=0.6,inner sep=0.0em,left = -0.2mm of m1]{$M\times$};
\node[diamond,draw=black,fill=black!10,inner sep=0.0em,below right= 0.2mm and 0.2mm of v1,scale=0.67] (m1sink) {$m_1$};
\node[right = 0mm of m1sink,inner sep=0.0em,scale=0.5] {$\times M$}; 
\node[rectangle,draw=black,fill=black!10,above=0.5cm of u2, inner sep=0.1em,scale=0.67] (m2)  {$m_2$@$t+2$};
\draw[->] (m2) -- (u2);
\node[scale=0.6,inner sep=0.0em,left = -0.2mm of m2]{$M\times$};
\node[diamond,draw=black,fill=black!10,inner sep=0.0em,below right= 0.2mm and 0.2mm of v2,scale=0.67] (m2sink) {$m_2$};
\node[right = 0mm of m2sink,inner sep=0.0em,scale=0.5] {$\times M$}; 
\node[rectangle,draw=black,fill=black!10,above=0.5cm of u3, inner sep=0.1em,scale=0.67] (m3)  {$m_3$@$t+2$};
\draw[->] (m3) -- (u3);
\node[scale=0.6,inner sep=0.0em,left = -0.2mm of m3]{$M\times$};
\node[diamond,draw=black,fill=black!10,inner sep=0.0em,below right= 0.2mm and 0.2mm of v3,scale=0.67] (m3sink) {$m_3$};
\node[right = 0mm of m3sink,inner sep=0.0em,scale=0.5] {$\times M$}; 
\draw[->,dashed] (8.5,5.8) -- (u4);
\node[rectangle,draw=black,fill=black!10,above=0.35cm of u5, inner sep=0.1em,scale=0.67] (m5)  {$m_{M-1}$@$t+2$};
\draw[->] (m5) -- (u5);
\node[scale=0.6,inner sep=0.0em,left = -0.2mm of m5]{$M\times$};
\node[diamond,draw=black,fill=black!10,inner sep=0.0em,below right= 0.2mm and 0.2mm of v5,scale=0.4] (m5sink) {$m_{M-1}$};
\node[right = 0mm of m5sink,inner sep=0.0em,scale=0.5] {$\times M$}; 
\node[rectangle,draw=black,fill=black!10,above=0.4cm of u6, inner sep=0.1em,scale=0.67] (m6)  {$m_{M}$@$t+2$};
\draw[->] (m6) -- (u6);
\node[scale=0.6,inner sep=0.0em,right = -0.2mm of m6]{$\times M$};
\node[diamond,draw=black,fill=black!10,inner sep=0.0em,below right= 0.2mm and 0.2mm of v6,scale=0.4] (m6sink) {$m_{M-1}$};
\node[right = 0mm of m6sink,inner sep=0.0em,scale=0.5] {$\times M$}; 

\end{tikzpicture}
\caption{An $(M,t)$-Backfire, where $M\in\N_{\geq 1}$ is some large number and $t\in\N_{\geq 0}$ is a time-step. Circles, rectangles and diamonds are resp. vertices, sources and sinks. After symbol @ is the source's starting time.
The edges that are plainly depicted (or dashed) have capacity one and fixed-delay zero. The filter is depicted in Fig. \ref{fig:filter}. Let priority $\succ$ satisfy $r_i\succ m_j\succ b_k\succ x$, for any $i,j,k$ defined in the figure.
One can connect to any digraph from the trigger to the bomb, if the minimum delay from $w_2$ to $w_3$ is one. \emph{Agent $x$ gets heavily delayed on the bomb if and only if he uses the trigger on time $t$.}}\label{fig:backfire}
\end{figure}
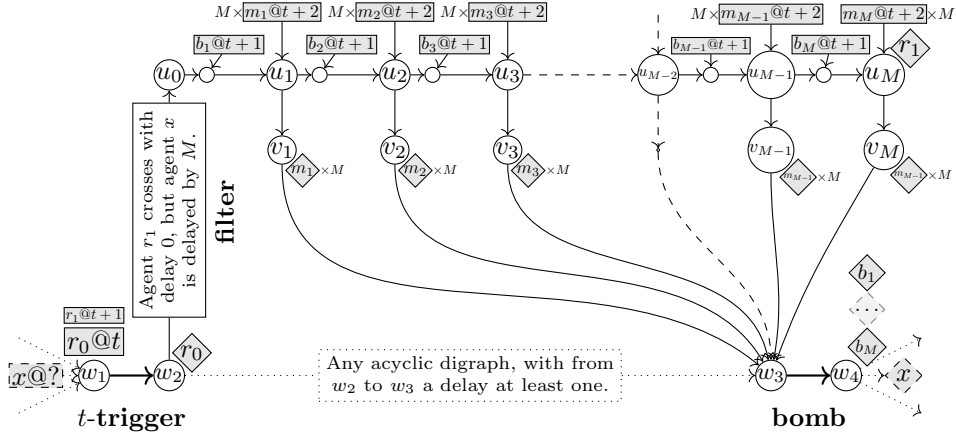

%, ,
\begin{definition}
An $(M,t)$-Backfire is a piece of \frog, defined as in Fig. \ref{fig:backfire}. %\vspace*{-2mm}
\end{definition}
\begin{lemma}\label{lem:backfire}
In an $(M,t)$-Backfire, if agent $x$ arrives on the $t$-trigger at time $t$,\\ then on the bomb, $M$ agents arrive at time $t+1$, and massively delay agent $x$.\\ Otherwise, if $x$ arrives at a different time, then this Backfire does not delay $x$.
Furthermore, the backfire contains $\Theta(M)$ vertices and $\Theta(M^2)$ agents.  %\vspace*{-1mm}
\end{lemma}
\begin{proof}[Lemma \ref{lem:backfire}]
If agent $x$ does not trigger on time $t$, 
then agent $r_1$ makes every agent $b_i$ wait one step.
Hence, every agent $b_i$ collides on $u_i$ at $t+2$ with $M$ agents $m_i$ who have priority.
Agents $b_i$ finally arrive on $w_3$ at time $t+1+M$, way too late to delay anyone (assuming large $M$).
If agent $x$ triggers on time $t$, then he gets queued after agent $r_0$, and agent $r_1$ has to wait one step. Then agent $r_1$ arrives too late on vertices $u_i$ to delay any agent $b_i$. Consequently, agents $b_i$ arrive on $u_i$ one step before $m_i$, don't get delayed, and arrive on $w_3$ at $t+1$. \qed
\end{proof}

\begin{definition}
An $M$-Backfire is a sequence of $(M,t)$-Backfires, for $0\leq t\leq M$, that share the same trigger-edge and bomb-edge.
Agent $r_0$ is removed everywhere but for $t=0$, because for $t\geq 1$, its role in the $(M,t)$-Backfire is played by $r_1$ from the $(M,t-1)$-Backfire. 
(There is one $r_1$ per-time-step in $[0,M]$.)%\vspace*{-1mm}
\end{definition}
\begin{lemma}\label{lem:backfire-universal}
With an $M$-Backfire, agent $x$ is massively delayed on the bomb-edge (assuming large $M$), if and only if he crosses the trigger-edge (anytime in $[0,M]$).\\ Furthermore, the backfire contains $\Theta(M^2)$ vertices and $\Theta(M^3)$ agents.  %\vspace*{-1mm}
\end{lemma}
\begin{proof}[Lemma \ref{lem:backfire-universal}]
Assume that agent $x$ triggers on time $t$; all subsequent agents $r_1$ get delayed by one: an $(M,t')$-Backfire gets triggered for every $t'\geq t$.\qed
\end{proof}
\begin{figure}[!t]
%\vspace*{0.5cm}
\centering
\begin{tikzpicture}
%\node[draw=white] at (12,0){};
%v_i nodes, every two 
\node[circle,draw=black, inner sep=0.2mm] (u0) at (1,0) {$u_0$};
\foreach \x in {0,2,...,8}
{
	\pgfmathtruncatemacro{\labelu}{\x/2};
	\pgfmathtruncatemacro{\labelv}{1+(\x/2)};
	\node[circle,draw=black, inner sep=0.2mm] 
		(nv\labelv) at (1.5+\x,-0.67) {$\overline{v_{\labelv}}$};
	\node[circle,draw=black, inner sep=0.2mm] 
		(v\labelv) at (1.5+\x,+0.67) {$v_{\labelv}$};
	\draw[->] (u\labelu) -- (nv\labelv);
	\draw[->] (u\labelu) -- (v\labelv);
	
	\node[circle,draw=black, inner sep=0.2mm] 
		(nw\labelv) at (2.5+\x,-0.67) {$\overline{w_{\labelv}}$};	
	\node[circle,draw=black, inner sep=0.2mm] 
		(w\labelv) at (2.5+\x,+0.67) {$w_{\labelv}$};
	\draw[->] (nv\labelv) -- node[above=2.5mm,scale=0.75] {\emph{don't take}} (nw\labelv);
	\draw[->] (v\labelv) -- node[below=1.3mm,scale=0.75] {\emph{take $\varphi_{\labelv}$}} (w\labelv);
		
	\node[circle,draw=black, inner sep=0.2mm] 
		(u\labelv) at (3.0+\x,0) {$u_{\labelv}$};
	\draw[->] (nw\labelv) -- (u\labelv);
	\draw[->] (w\labelv) -- (u\labelv);
	
}
%triggers, bombs, one agent.
\draw[->,line width=1pt]	(nv2) -- node[below=0.2cm]{\textbf{trigger~~~~}} (nw2);
\draw[->,line width=1pt]	(nv4) -- node[below=0.2cm]{\textbf{bomb}} (nw4);
\draw[->,line width=1pt]	(nv5) -- node[below=0.2cm]{\textbf{bomb}} (nw5);
\node[rectangle,draw=black,fill=black!10,dashed,left= 1mm of u0]{$x$};
\node[diamond,draw=black,fill=black!10,dashed,right= 1mm of u5, inner sep=0.5mm]{$x$};
%Backfire
\node[rectangle,draw=black,thick] (backfire) at (6,-1.4) {$M$-Backfire};
\draw[->,thick] (nw2) edge[out=270,in=180] (backfire);
\draw[->,thick] (backfire) edge[out=0,in=240](nv4);
\draw[->,thick] (backfire) edge[out=0,in=240] (nv5);
\node[] (blabla) at (6,-2) {``Edges $\{\varphi_2,\varphi_4\}$ and $\{\varphi_2,\varphi_5\}$ shall be covered, or a backfire will heavily delay agent $x$.''};
\end{tikzpicture}
%\vspace*{-0.5cm}
\caption{From \minvc~ (degrees bounded above by $3$) to \textsc{Frog/Br/Dec}. 
Circles, squares and diamonds depict respectively vertices, sources and sinks for \frog.
Let $\eta=|\CV|$ and observe that the starting size is in $\Theta(\eta)$.
In the depiction, $\eta=5$.
The idea is a correspondence between $\CW\in 2^{\CV}$ and path $\pi_x$ decided by agent $x$:
\emph{taking edge $(v_i,w_i)$ in path $\pi_x$ amounts to take vertex $\varphi_i$ in subset $\CW$.}
Every edge is associated with $(c_e,d_e)=(1,1)$, but edges $(v_i,w_i)$ with $(1,2)$, and edges $(\overline{v_i},\overline{w_i})$ when it's a trigger with $(1,0)$ (because agent $r_1$ already makes $x$ wait one step).
Consequently going up always takes two steps,
and going down one step if it's not a backfired edge.. 
Hence a vertex cover $\CW$ of size $k$ corresponds to a path $\pi_x$ with length $3\eta+k$.
So, threshold $\kappa$ in \minvc~is reduced to $\kappa'=3\eta+k$ in \textsc{Frog/Br/Dec}.
For every edge $\{\varphi_i,\varphi_j\}$ ($i<j$) in \minvc,
we introduce an $M$-Backfire with trigger $(\overline{v_i},\overline{w_i})$ and bomb $(\overline{v_j},\overline{w_j})$, in order to heavily punish $x$ for not taking $\varphi_i$  \emph{and} $\varphi_j$. The backfire splits the provided punishement between up to three neighbors.}\label{fig:minvc2br}
\end{figure}
%\vspace*{-0.5cm}
\begin{proof}[Theorem \ref{th:NPcomplete}]
Membership in class NP follows from Th. \ref{th:map}. We show NP-hardness by starting the reduction from \emph{decision} problem \minvc~\cite{karp1972reducibility,garey1979computers} that asks, given graph $\CG=(\CV,\CE)$ and threshold $\kappa\in\N_{\geq 0}$, whether there is a subset $\CW\subseteq\CV$ such that\quad$\forall\{\varphi_1,\varphi_2\}\in\CE,~\varphi_1\in\CW\mbox{ or }\varphi_2\in\CW$,\quad and\quad $|\CW|\leq\kappa$.
Recall that problem \minvc~is NP-complete even for degrees bounded above by three \cite{Garey:1974:SNP}, which we assume here.  
We build the \frog~depicted in Fig. \ref{fig:minvc2br}.
The reduction's validity is by construction (see the figure's caption).
Taking $M= 6\eta$ is sufficient. Since there are $\Theta(\eta)$ edges in $\CV$, the reduction makes $\Theta(\eta^3)$ vertices and $\Theta(\eta^4)$ agents, which is polynomial.
\qed
%($\Rightarrow$)
%There is a correspondance between $\CW\in 2^{\CV}$ and the path decided by agent $x$:
%taking edge $(v_i,w_i)$ in the path amounts to take vertex $\varphi_i$ in subset $\CW$.
%If there is a vertex cover $|\CW|\leq\kappa$, the corresponding path has length $n+|\CW|\leq n+\kappa=\kappa'$.
%($\Leftarrow$) If there is a path with delay $|\pi_i|\leq n+\kappa$, then it avoided avery bomb, and it corresponds to a vertex cover $|\CW|\leq\kappa$).
\end{proof}

\begin{proof}[Theorem \ref{th:APXhard}]
Starting from the \emph{optimization} version of problem \minvc~ where one must find 
$\CW^\ast\in\mbox{arg}\min_{\CW\subseteq\CV}\left\{|\CW| \mid \forall\epsilon\in\CE,~\CW\cap\epsilon\neq\emptyset\right\}$ (forget about $\kappa$ and $\kappa'$), 
the same reduction as for Th. \ref{th:NPcomplete}  is also an $L$-reduction\footnote{An $L$-reduction is a poly.-time reduction in NPO, which conserves approximations.} \cite{papadimitriou1991optimization,crescenzi1997short}, which we show by exhibiting functions $f,g$ and constants $\alpha,\beta$.

Recall that \emph{optimization} problem \minvc~is APX-complete even for degrees bounded above by three \cite{papadimitriou1991optimization,alimonti1997hardness}, which we still assume.
The correspondences $f$ and $g$ are depicted in the caption of Fig. \ref{fig:minvc2br}.
Given a \minvc~instance $\CI$, one has $\mbox{OPT}_{\frog}(f(\CI))\leq\alpha\mbox{OPT}_{\textsc{VC}}(\CI)$ for $\alpha=10$ and $|\CW|\leq\beta C_i(\pi_x)$ where $\CW=g(\CI,\pi_x)$  for $\beta=1$. Indeed, for the former, observe that a vertex can cover at most three edges; hence $\frac{\eta}{3}\leq \mbox{OPT}_{\textsc{VC}}(\CI)$. Correspondence $3\eta+\mbox{OPT}_{\textsc{VC}}=\mbox{OPT}_{\frog}$ then yields $\alpha=10$. The later comes from $k\leq 3\eta+k$. Consequently, this is an $L$-reduction and then, optimization problem \textsc{Frog/Br/Opt} is APX-hard. \qed
\end{proof}

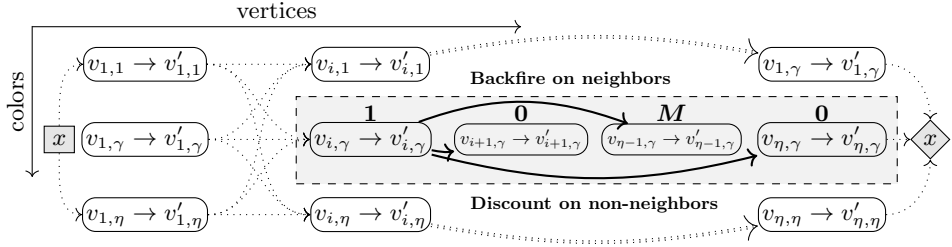
\begin{figure}[!t]
\hspace*{-4mm}
\begin{tikzpicture}
\pgfmathtruncatemacro{\mx}{3};
\pgfmathtruncatemacro{\my}{1};
\pgfmathtruncatemacro{\px}{2};
%rectangle?
\draw[dashed,fill=black!05] (-0.33*\mx,0.57*\my) rectangle (3.5*\px,-0.59*\my);
%FIRST NODE:
\node[rectangle,draw=black,rounded corners=5pt,inner sep=0.5mm] (5)
at (0,0) {$v_{i,\gamma}\rightarrow v_{i,\gamma}'$};
%five nodes on the left
\node[rectangle,draw=black,rounded corners=5pt,inner sep=0.5mm] (6)
at (0,-\my) {$v_{i,\eta}\rightarrow v_{i,\eta}'$};
\node[rectangle,draw=black,rounded corners=5pt,inner sep=0.5mm] (3)
at (-\mx,-\my) {$v_{1,\eta}\rightarrow v_{1,\eta}'$};
\node[rectangle,draw=black,rounded corners=5pt,inner sep=0.5mm] (2)
at (-\mx,0) {$v_{1,\gamma}\rightarrow v_{1,\gamma}'$};
\node[rectangle,draw=black,rounded corners=5pt,inner sep=0.5mm] (1)
at (-\mx,+\my) {$v_{1,1}\rightarrow v_{1,1}'$};
\node[rectangle,draw=black,rounded corners=5pt,inner sep=0.5mm] (4)
at (0,+\my) {$v_{i,1}\rightarrow v_{i,1}'$};
\path[dotted,draw=black]
	(1)edge[out=0,in=180,->](4) 
	(1)edge[out=0,in=180,->](5)
	(1)edge[out=0,in=180,->](6)
	(2)edge[out=0,in=180,->](4) 
	(2)edge[out=0,in=180,->](5)
	(2)edge[out=0,in=180,->](6)
	(3)edge[out=0,in=180,->](4) 
	(3)edge[out=0,in=180,->](5)
	(3)edge[out=0,in=180,->](6);
%sequence right of center
\node[rectangle,draw=black,rounded corners=5pt,inner sep=1mm,scale=0.75] 
(b1) at (\px,0) {$v_{i+1,\gamma}\rightarrow v_{i+1,\gamma}'$};
\node[rectangle,draw=black,rounded corners=5pt,inner sep=1mm,scale=0.75] 
(b2) at (2*\px,0) {$v_{\eta-1,\gamma}\rightarrow v_{\eta-1,\gamma}'$};
\node[rectangle,draw=black,rounded corners=5pt,inner sep=0.5mm] 
(b3) at (3*\px,0) {$v_{\eta,\gamma}\rightarrow v_{\eta,\gamma}'$};
%last corners right
\node[rectangle,draw=black,rounded corners=5pt,inner sep=0.5mm] (c2)
at (3*\px,-\my) {$v_{\eta,\eta}\rightarrow v_{\eta,\eta}'$};
\node[rectangle,draw=black,rounded corners=5pt,inner sep=0.5mm] (c1)
at (3*\px,+\my) {$v_{1,\gamma}\rightarrow v_{1,\gamma}'$};
\draw[] (4) edge[->,dotted,double,out=10,in=170] (c1);
\draw[] (6) edge[->,dotted,double,out=350,in=190] (c2);
%Backfires/Discounts
\path[]
	(5) edge[->,thick,out=20,in=160] 
		node[above right=1mm and -8mm,scale=0.75]{\textbf{Backfire on neighbors}} (b2)
	(5) edge[->,thick,out=-10,in=-170] (b1)
	(5) edge[->,thick,out=-14,in=-166] 
		node[below=1.5mm,scale=0.75,align=center]{\textbf{Discount on non-neighbors}} (b3);
%Legendes:
\draw[] (-1.5*\mx,1.5*\my) 
edge[->,thin] node[above]{vertices} (1*\px,1.5*\my);
\draw[] (-1.5*\mx,1.5\my) 
edge[->,thin] node[above,rotate=90]{colors} (-1.5*\mx,-0.5*\my);
%agent
\node[rectangle,draw=black,fill=black!10](xs) at (-1.38*\mx,0) {$x$};
\node[diamond,draw=black,fill=black!10,inner sep=0.5mm](xd) at (3.72*\px,0) {$x$};
\draw[] (xs) edge[->,dotted,out=90,in=180] (1);
\draw[] (xs) edge[->,dotted,out=270,in=180] (3);
\draw[] (c1) edge[->,dotted,out=0,in=90] (xd);
\draw[] (b3) edge[->,dotted,out=0,in=180] (xd);
\draw[] (c2) edge[->,dotted,out=0,in=270] (xd);
%costs
\node[] at (0,0.35*\my) {\bf 1};
\node[] at (\px,0.35*\my) {\bf 0};
\node[] at (2*\px,0.35*\my) {$\bm{M}$};
\node[] at (3*\px,0.35*\my) {\bf 0};
\end{tikzpicture}
\caption{Sketch of reduction from \mincoloring~to \textsc{Frog/Br/Opt}.
Agent $x$'s paths correspond to deciding a color $\gamma$ for each vertex $i$.
Depicted in the gray rectangle, the first time $x$ chooses color $\gamma$ is on vertex $i$ (first time on the line). Then (1) $x$ waits one step on edge $(v_{i,\gamma},v_{i,\gamma}')$ (because of an agent $r_1$). Then on the same line, (2) neighbors in $\CG$ are Backfired (can't put the same color on a neighbor) and non-neighbors are discounted to delay zero (by heavily delaying agents $r_1$ and disarming their eventual backfires). Transit edges (dotted) have delay one to allow for backfires to work. Hence, a valid coloring of size $k$ would correspond to a path $\pi_x$ of length $\eta+k$ which does not enable to find $\beta$ for an $L$-reduction.
Therefore, to bring the correspondence back from affine to linear, we technically multiply all the costs by $M$, with $M$ times more agents and vertices, but not on the dotted edges. }
\label{fig:chromatic}
\end{figure}
\begin{proof}[Sketch, Th. \ref{th:inapprox}]
Problem \mincoloring, 
given graph $\CG=(\CV,\CE)$, asks a coloring of $\CG$,
i.e. a partition of $\CV$ into disjoint sets $V_1,V_2,\ldots,V_k$ such that each $V_i$ is an independent set of $\CG$ (no edges in $\CG[V_i]$),
with \emph{minimum chromatic number} $k=\chi(\CG)$.
Let $\eta=|\CV|$.
It is known that  whatever $\varepsilon>0$, approximating $\chi(\CG)$ within $\eta^{1-\varepsilon}$ is NP-hard \cite{feige1996zero,zuckerman2006linear}.
The idea of the reduction is in Fig. \ref{fig:chromatic},
and (with $M=3\eta$) involves $\Theta(\eta^6)$ vertices and $\Theta(\eta^7)$ agents. Consequently, better approximation ratios than $|V|^{\frac{1}{6}-\varepsilon}$ 
or ${n^{\frac{1}{7}-\varepsilon}}$  contradict intractable ratio $\eta^{1-\varepsilon}$ from \cite{feige1996zero}.\qed
\end{proof}

\section{The Complexity of Pure Nash Equilibria}

In this section, we first observe that the verification problem \textsc{Frog/NE/Verif} is coNP-complete.
%, maybe pushing problem \textsc{Frog/NE/Exist}  out of class NP. Indeed, we then
Then, we completely characterize the complexity of the existence problem as \emph{complete} for the second level of PH\footnote{Class $\Sigma_2^P$ are the problems that nest a coNP problem inside an NP problem.\\ Only very small sizes ($\lessapprox 10$) of such problems can usually be practically addressed.}. %Rationality is vain.

\begin{theorem}\label{th:coNP}
Problem \textsc{Frog/NE/Verif} is coNP-complete \cite[Almost]{werth2014atomic}$^{(g)}$.
\end{theorem}

\begin{proof}[Theorem \ref{th:coNP}]
A deviation is a no-certificate verifiable in polynomial-time by Alg. \ref{alg:map}, hence this problem is inside class coNP.
A proof with bottleneck objectives lies in \cite[Cor. 4]{werth2014atomic}, and the authors claim \cite[Sec. 7]{werth2014atomic} that one can obtain NP-hardness for sum-objectives in the same way. We confirm that claim since the same reduction as for Th. \ref{th:NPcomplete} holds here. \qed
\end{proof}

\begin{theorem}\label{th:sigma2}
Problem \textsc{Frog/NE/Exist} is $\Sigma_2^P$-complete.
\end{theorem}
\vspace*{-5mm}
\begin{figure}[t]
\centering
\begin{tikzpicture}
\node[rectangle,draw=black,inner sep=2mm] (v0d) at (0,-1) {$\CV_{1,1}$};
\node[rectangle,draw=black,inner sep=2mm] (v0u) at (0,+1) {$\CV_{1,0}$};
\node[rectangle,draw=black,inner sep=2mm] (v1d) at (2,-1) {$\CV_{2,1}$};
\node[rectangle,draw=black,inner sep=2mm] (v1u) at (2,+1) {$\CV_{2,0}$};
\node[rectangle,draw=black,inner sep=2mm] (v2d) at (4,-1) {$\CV_{i,1}$};
\node[rectangle,draw=black,inner sep=2mm] (v2u) at (4,+1) {$\CV_{i,0}$};
\node[rectangle,draw=black,inner sep=2mm] (v3d) at (6,-1) {$\CV_{|I|,1}$};
\node[rectangle,draw=black,inner sep=2mm] (v3u) at (6,+1) {$\CV_{|I|,0}$};
%agent theta middle
\node[circle,draw=black,inner sep=0.5mm] (m1) at (-1.5,0) {};
\node[circle,draw=black,inner sep=0.5mm] (m2) at (0.5,0) {};
\node[circle,draw=black,inner sep=0.5mm] (m3) at (2.5,0) {};
\node[circle,draw=black,inner sep=0.5mm] (m4) at (4.5,0) {};
\node[circle,draw=black,inner sep=0.5mm] (m5) at (6.5,0) {};
%agent theta up
\node[circle,draw=black,inner sep=0.5mm] (u1) at (-1.0,0.5) {};
\node[circle,draw=black,inner sep=0.5mm] (u2) at (1.0,0.5) {};
\node[circle,draw=black,inner sep=0.5mm] (u3) at (3.0,0.5) {};
\node[circle,draw=black,inner sep=0.5mm] (u4) at (5.0,0.5) {};
%agent theta down
\node[circle,draw=black,inner sep=0.5mm] (d1) at (-1.0,-0.5) {};
\node[circle,draw=black,inner sep=0.5mm] (d2) at (1.0,-0.5) {};
\node[circle,draw=black,inner sep=0.5mm] (d3) at (3.0,-0.5) {};
\node[circle,draw=black,inner sep=0.5mm] (d4) at (5.0,-0.5) {};
%agent theta arrows external
\path[draw=black,thick]
	(m1) edge[->] (u1)
	(m1) edge[->] (d1)
	(m2) edge[->] (u2)
	(m2) edge[->] (d2)
	(m3) edge[->] (u3)
	(m3) edge[->] (d3)
	(m4) edge[->] (u4)
	(m4) edge[->] (d4);
%agent theta arrows internal
\path[draw=black]
	(u1) edge[->] (m2)
	(d1) edge[->] (m2)
	(u2) edge[->] (m3)
	(d2) edge[->] (m3)
	(u3) edge[->] (m4)
	(d3) edge[->] (m4)
	(u4) edge[->] (m5)
	(d4) edge[->] (m5);
%agent theta source and sink
\node[rectangle,draw=black,fill=black!10,left=0mm of m1]{$\theta$};
\node[diamond,draw=black,fill=black!10,right=0mm of m5,inner sep=0.2mm]{$\theta$};
%backfires up
\path[draw=black,thick]
	(u1) edge[->,out=90,in=180] node[left]{$\bm{b}$}(v0u)
	(u2) edge[->,out=90,in=180] node[left]{$\bm{b}$}(v1u)
	(u3) edge[->,out=90,in=180] node[left]{$\bm{b}$}(v2u)
	(u4) edge[->,out=90,in=180] node[left]{$\bm{b}$}(v3u);
%backfires down
\path[draw=black,thick]
	(d1) edge[->,out=270,in=180] node[left]{$\bm{b}$}(v0d)
	(d2) edge[->,out=270,in=180] node[left]{$\bm{b}$}(v1d)
	(d3) edge[->,out=270,in=180] node[left]{$\bm{b}$}(v2d)
	(d4) edge[->,out=270,in=180] node[left]{$\bm{b}$}(v3d);
%Agent x
\node[circle,draw=black,inner sep=0.5mm] (mzerox) at (-2.25,0) {};
\node[rectangle,draw=black,fill=black!10,left=0mm of mzerox]{$x$};
\node[circle,draw=black,inner sep=0.5mm] (mfinix) at (7.4,0) {};

%FIGURE 2
\node[rectangle,draw=black,inner sep=1mm] (mz) at (9,0) {Figure \ref{fig:counter}};
%Agent x backfires
\path[draw=black,thick]
	(mfinix) edge[->] node[below right=0.5mm and -2mm]{$\bm{b_{>1+n+\kappa}}$}(mz);
%many double curved dotted edges everywhere
\path[draw=black!50,dotted]
		(mzerox) edge[double,->,out=90,in=160] (v0u)
		(mzerox) edge[double,->,out=270,in=200] (v0d)
%four edges:
		(v0u) edge[double,->] (v1u)
		(v0u) edge[double,->] (v1d)
		(v0d) edge[double,->] (v1u)
		(v0d) edge[double,->] (v1d)
%four edges:
		(v1u) edge[double,->] (v2u)
		(v1u) edge[double,->] (v2d)
		(v1d) edge[double,->] (v2u)
		(v1d) edge[double,->] (v2d)
%four edges:
		(v2u) edge[double,->] (v3u)
		(v2u) edge[double,->] (v3d)
		(v2d) edge[double,->] (v3u)
		(v2d) edge[double,->] (v3d)
%two last edges:
		(v3u) edge[double,->,out=0,in=90] (mfinix)
		(v3d) edge[double,->,out=0,in=270] (mfinix);	
%agent x		
\node[diamond,draw=black,fill=black!10,inner sep=0.2mm,right=0mm of mfinix]{$x$};
\end{tikzpicture}

\caption{Reduction from  \maxminvc~to the complement of \textsc{Frog/NE/Exist}. Generalizes the reduction in Fig. \ref{fig:minvc2br}. Each box $\CV_{i,j}$ is made as in Fig. \ref{fig:minvc2br}. We create a universally indifferent agent $\overline{\theta}$ who can early decide between two paths for every $i\in I$, : one backfires $\CV_{i,0}$'s entry and the other backfires $\CV_{i,1}$'s entry. Agent $\overline{\theta}$ models function $\theta$ by blocking early the entries to $\CV_{i,0}\mbox{ xor to }\CV_{i,1}$ with the backfires $\bm{b}$. Plain edges have capacity and cost $(c_e,d_e)=(1,0)$. Dotted edges have $(c_e,d_e)=(1,1)$ but the two first ones $(1,2)$, to let $\overline{\theta}$ run in front of $x$. Then agent $x$ decides a path through what corresponds to subgraph $\CG^{(\theta)}$. Agent $x$ sends backfires to Fig. \ref{fig:counter} if and only if he reaches his sink after time $1+n+\kappa$.}
\label{fig:minvc2exist}
\end{figure}
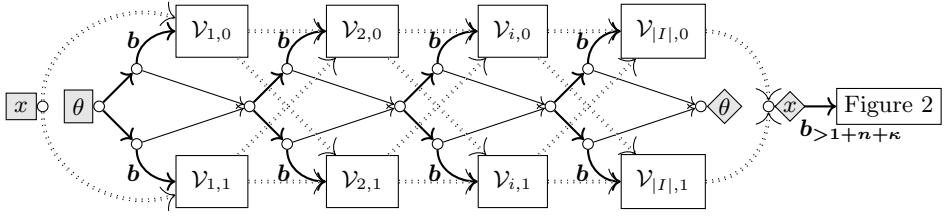
\begin{proof}[Theorem \ref{th:sigma2}]
This problem is in class $\Sigma_2^P$. Indeed, yes-instances admit a certificate verifiable by an NP-oracle: by guessing the right strategy-profile, according to Th. \ref{th:coNP}, one can use an NP-oracle to verify that it is a PNE. The $\Sigma_2^P$-hardness proof below generalizes the reduction introduced for Th. \ref{th:NPcomplete}.

In Fig. \ref{fig:minvc2exist}, we reduce decision problem \maxminvc~to the complement of  \textsc{Frog/NE/Exist}.
Given set of indices $I$, the vertices of graph $\CG=(\CV,\CE)$
partition into $\CV=\bigcup_{i\in I} \CV_{i,0}\cup\CV_{i,1}$.
Given function $\theta:I\rightarrow\{0,1\}$ (i.e. $2^{|I|}$ possibilities),
let $\CG^{(\theta)}$ denote the graph restricted to vertices $\CV^{(\theta)}=\bigcup_{i\in I}\CV_{i,t(i)}$. 
Problem \maxminvc, given threshold $\kappa\in\N_\geq 0$, asks whether:
\vspace*{-2mm}
\begin{eqnarray}
\forall \theta:I\rightarrow\{0,1\},\quad
\exists \CW\subseteq \CV^{(\theta)},\quad
\CW\mbox{ vertex-covers }\CG^{(\theta)}\mbox{ and }|\CW|\leq\kappa,
\label{eq:1}
\end{eqnarray}
and is $\Pi_2^P$-complete (i.e. co-$\Sigma_2^P$-complete);
 co-\textsc{Frog/NE/Exist} asks whether:
 \vspace*{-1mm}
\begin{eqnarray}
\forall \bm{\pi}\in\CP,\quad
\mbox{There exists an individual deviation from $\bm{\pi}$.}
\label{eq:2}
\end{eqnarray}

[Eq. (\ref{eq:1}) $\Rightarrow$ Eq. (\ref{eq:2})]
Whatever the choices of agent $\overline{\theta}$,
if the strategy of agent $x$ costs more than $C_i>1+n+\kappa$, then he can deviate and improve, because of Eq. (\ref{eq:1}); otherwise, now assuming that $x$'s strategy is a best-response, then he reaches his sink before time $1+n+k$ (because Eq. (\ref{eq:1})) and does not disable the example from Fig. \ref{fig:counter}, which remains unstable: there is a deviation.

[not Eq. (\ref{eq:1}) $\Rightarrow$ not Eq. (\ref{eq:2})] If there exists a function $\theta$, then we position agent $\overline{\theta}$ as such. Then the best-response of agent $x$ makes him reach his sink after time $1+n+\kappa$. Consequently, Fig. \ref{fig:counter} is disabled: we have a PNE.\qed
\end{proof}

\vspace*{-1mm}
\section{The Price of GPS}
\vspace*{-2mm}

Previous sections show how strong an assumption rationality is.
Instead, we propose a model inspired by GPS personal navigation assistants: agents retrieve instantaneous traffic data to recompute shortest paths at each crossroad. %Billions were invested in the making of the GPS system, and it is interesting to study the efficiency of its use.

We introduce a \emph{GPS-agent} as an agent who
at each vertex (between two time steps) recalculates a shortest path according to the fixed delays $d_e$ plus congestion $\lfloor\frac{|q_e|}{c_e}\rfloor$ of the past step. In place of PNE, let $\CO\subseteq\CP$ be the set of strategy-profiles that can be obtained by GPS-agents. We study the worst-case ratio to coordination, defined for one \frog~as the \emph{Price-of-GPS (navigation)}:
\begin{eqnarray*}
\mbox{PoGPS}
&=&
\frac
{\max_{\bm{\pi}'\in\CO }\left\{~C(\bm{\pi}')~\right\}}
{\min_{\bm{\pi}\in\CP }\left\{~C(\bm{\pi})~\right\}},
\end{eqnarray*}
where $C(\bm{\pi})=\sum_{i\in N}C_i(\bm{\pi})$.
For a family of \frog s, PoGPS is the supremum of every PoGPS therein. 
As shown in Fig. \ref{fig:procras}, a first negative result follows:
\begin{figure}[t]
\centering
\begin{tikzpicture}
% nodes north-east
\node[circle,draw=black,inner sep=0.2mm,scale=0.75] (u00) at (2.0,0.5) {$u_{00}$};
\node[circle,draw=black,inner sep=0.2mm,scale=0.75] (u01) at (1.5,1.0) {$u_{01}$};
\node[circle,draw=black,inner sep=0.2mm,scale=0.75] (v00) at (2.5,1.0) {$v_{00}$};
\node[circle,draw=black,inner sep=0.2mm,scale=0.75] (v01) at (2.0,1.5) {$v_{01}$};
%nodes north-west
\node[circle,draw=black,inner sep=0.2mm,scale=0.75] (u11) at (-2.0,0.5) {$u_{11}$};
\node[circle,draw=black,inner sep=0.2mm,scale=0.75] (u10) at (-1.5,1.0) {$u_{10}$};
\node[circle,draw=black,inner sep=0.2mm,scale=0.75] (v11) at (-2.5,1.0) {$v_{11}$};
\node[circle,draw=black,inner sep=0.2mm,scale=0.75] (v10) at (-2.0,1.5) {$v_{10}$};
%nodes south-west
\node[circle,draw=black,inner sep=0.2mm,scale=0.75] (u20) at (-2.0,-0.5) {$u_{20}$};
\node[circle,draw=black,inner sep=0.2mm,scale=0.75] (u21) at (-1.5,-1.0) {$u_{21}$};
\node[circle,draw=black,inner sep=0.2mm,scale=0.75] (v20) at (-2.5,-1.0) {$v_{20}$};
\node[circle,draw=black,inner sep=0.2mm,scale=0.75] (v21) at (-2.0,-1.5) {$v_{21}$};
%nodes south-east
\node[circle,draw=black,inner sep=0.2mm,scale=0.75] (u31) at (2.0,-0.5) {$u_{31}$};
\node[circle,draw=black,inner sep=0.2mm,scale=0.75] (u30) at (1.5,-1.0) {$u_{30}$};
\node[circle,draw=black,inner sep=0.2mm,scale=0.75] (v31) at (2.5,-1.0) {$v_{31}$};
\node[circle,draw=black,inner sep=0.2mm,scale=0.75] (v30) at (2.0,-1.5) {$v_{30}$};
%edges
\path[]
%edges east
	(u31) edge[-{>[width=7pt]}, line width = 2pt] (u00)
	(v31) edge[-{>[width=7pt]}, line width = 2pt] (v00)
%edges north east
	(u00) edge[->] (u01)
	(u00) edge[-{>[width=7pt]}, line width = 2pt] (v01)
	(v00) edge[->] (u01)
	(v00) edge[->] (v01)
%edges north
	(u01) edge[->] (u10)
	(v01) edge[-{>[width=7pt]}, line width = 2pt] (v10)
%edges north west
	(u10) edge[->] (u11)
	(u10) edge[->] (v11)
	(v10) edge[->] (u11)
	(v10) edge[->] (v11)
%edges west
	(u11) edge[->] (u20)
	(v11) edge[->] (v20)
%edges south west
	(u20) edge[->] (u21)
	(u20) edge[->] (v21)
	(v20) edge[->] (u21)
	(v20) edge[->] (v21)
%edges south
	(u21) edge[->] (u30)
	(v21) edge[->] (v30)
%edges south east
	(u30) edge[-{>[width=7pt]}, line width = 2pt] (u31)
	(u30) edge[-{>[width=7pt]}, line width = 2pt] (v31)
	(v30) edge[->] (u31)
	(v30) edge[->] (v31);
%outer edges for sinks
\path[dotted]
	(v00) edge[->] (3.0,1.5)
	(v10) edge[->] (-2.5,2.0)
	(v20) edge[->] (-3.0,-1.5)
	(v30) edge[->] (2.5,-2.0);
%inner edges for sinks
\path[dotted]
	(u00) edge[->] (1.5,0.0)
	(u10) edge[->] (-1.0,0.5)
	(u20) edge[->] (-1.5,0.0)
	(u30) edge[->] (1.0,-0.5);
%current agents:
\node[diamond,draw=black,fill=black!30,above = 0mm  of u30,inner sep=-0.2mm] {$i_1$};
\node[diamond,draw=black,fill=black!30,below = 0mm  of u10,inner sep=-0.2mm] {$i_2$};
\node[diamond,draw=black,fill=black!30,left = 0mm of v20,inner sep=-0.2mm] {$o_1$};
\node[diamond,draw=black,fill=black!30,right = 0mm of v00,inner sep=-0.2mm] {$o_2$};
%current congestion:
\begin{scope}[on background layer]
%above down
	\draw[draw=black!10, fill=black!10] (-1.6,-1.2) rectangle (1.7,-0.8);
	\draw[draw=black!10, fill=black!10] (-1.7,1.2) rectangle (1.6,0.8);
%left right
	\draw[draw=black!10, fill=black!10] (2.3,-1.1) rectangle (2.7,1.2);
	\draw[draw=black!10, fill=black!10] (-2.3,-1.2) rectangle (-2.7,1.1);
%corners tambouille
	\draw[draw=black!10, fill=black!10,rotate=-45] (2.3,0.55) rectangle (2.7,1.1);
	\draw[draw=black!10, fill=black!10,rotate=-45] (-2.3,-0.55) rectangle (-2.7,-1.1);
	\draw[draw=black!10, fill=black!10,rotate=-45] (0.3,1.6) rectangle (0.9,2.0);
	\draw[draw=black!10, fill=black!10,rotate=-45] (-0.3,-1.6) rectangle (-0.9,-2.0);
\end{scope}
\end{tikzpicture}
\caption{\emph{Double cycle of infinite procrastination.}
%idea
The idea is that there is an inner-cycle and an equivalent outer-cycle.
Agents from a cycle have to go through the other cycle to reach their sink, but the information that they get from the other cycles does not discourage procrastination.
%graph 
Circles are nodes. Every edge $e$ has capacity $c_e=1$. The four edges in every corner have fixed-delay $d_e=0$, and the two from each corner to the next one, fixed-delay $d_e=1$.
%agents
There are two \emph{inner-agents} $i_1$ and $i_2$, with resp. sources $u_{11}$ and $u_{31}$, and a sink reachable instantly by the dotted edges from the outer cycle's vertices $v_{00}, v_{10},v_{20},v_{30}$. 
However, they can decide to stay on the inner-cycle $u_{0-}, u_{1-},u_{2-},u_{3-}$.
There are two \emph{outer-agents} $o_1$ and $o_2$, with resp. sources $v_{21}$ and $v_{01}$, and sinks reachable by the dotted edges from inner vertices $u_{00}, u_{10}, u_{20}, u_{30}$. 
However, they can decide to cycle on the outer-cycle $v_{0-}, v_{1-},v_{2-},v_{3-}$.
On the figure, we show w.l.o.g. current positions of the agents and the congestion from the last step in gray rectangles. The current choice faced by agent $i_1$ is depicted with double edges.}\label{fig:procras}
\end{figure}
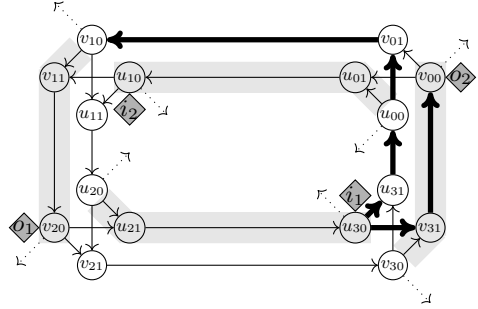
\begin{theorem}\label{th:gps-infinite}
Allowing walks$^{(i)}$, GPS-agents may cycle infinitely (Fig.\ref{fig:procras}). 
\end{theorem}
\begin{proof}
As depicted in Fig.\ref{fig:procras}, consider w.l.o.g. the end of a given time step, and the current choice faced by agent $i_1$. 
Straight outside shows congestion and is not better than taking the later exit at the next node.
Since every agent faces the same choice and the game is symmetric, it is possible to loop endlessly. \qed
\end{proof}
Following Th. \ref{th:gps-infinite}, 
we now focus on simple-paths and study the order of $\mbox{PoGPS}$.
\begin{theorem}
The Price of GPS Navigation is in $\Omega(|V|+n)$ as the number of vertices $|V|$ and the number of agents $n$ grow\footnote{For two variables $x,y$, Landau notation $f(x,y)\in\Omega(g(x,y))$ is defined as:\\ $\exists K\in\R_{>0},~\exists n_0\in\N_{\geq 0},~\forall x,y\geq n_0,~f(x,y)\geq Kg(x,y)$.}.
\end{theorem}
\begin{proof}
It suffices to generalize the double cycle of Fig. \ref{fig:procras} from 4 corners and agents, to a similar double cycle with more corners and agents. Then for every agent, while the shortest path has total-delay in $\Theta(1)$, the decided path can have total-delay in $\Theta(|V|)$ and $\Theta(n)$.\qed
\end{proof}

\vspace*{-2mm}
\section{Prospects}
\vspace*{-1mm}

The symmetric case seems usually well behaved \cite[Th.1]{hoefer2009competitive} and would be worth investigating.
Time expanded graphs, where one does the cross product of vertices and time or positions may yield an other beautiful approach.
A study on tie-breaking under FIFO is motivated by its importance in the proofs. 
Studying less extreme, average, or sub-cases would be appealing.
Extensive forms with decisions on each node \cite{cao2017ec} are a promising model.
%
%A practical prospect would be to make a range of real-life experiments on average case efficiency of GPS, or some measurements of pedestrian behaviors.
%
Finally, one could study the impact of new technologies of IP, like video streaming.

\textbf{Acknowledgments} 
I am grateful to the anonymous reviewers for their work.
Following recent practices, the reviews will be appended to a preprint.

\bibliographystyle{alphaCapitalization}
\bibliography{mybib}

\end{document}